\documentclass[12pt]{article}
\usepackage{amsmath,amsfonts,amssymb,amsthm}
\usepackage{xcolor}
\usepackage{graphicx}

\usepackage[colorlinks=true,linkcolor=black,citecolor=black,urlcolor=black,bookmarks,breaklinks=true]{hyperref}

\usepackage{enumitem}
\setitemize{itemsep=0.1pt}
\setenumerate{itemsep=0.1pt}
\usepackage{caption}
\usepackage{subcaption}
\usepackage{mathtools}

\mathtoolsset{showonlyrefs}

\newtheorem{proposition}{Proposition}[section]
\newtheorem{theorem}[proposition]{Theorem}
\newtheorem{corollary}[proposition]{Corollary}
\newtheorem{lemma}[proposition]{Lemma}
\newtheorem{remark}[proposition]{Remark}
\newtheorem{definition}[proposition]{Definition}
\newtheorem{example}[proposition]{Example}

\newcommand{\nc}{\newcommand}
\nc{\I}{{\mathbf 1}}

\nc{\bN}{{\mathbf N}}
\nc{\bM}{{\mathbf M}}
\nc{\cB}{{\mathcal B}}
\nc{\cL}{{\mathcal L}}
\nc{\R}{{\mathbb R}}
\nc{\N}{{\mathbb N}}
\nc{\Z}{{\mathbb Z}}
\nc{\bF}{{\mathbf F}}

\nc{\BP}{\mathbb{P}}
\nc{\BE}{\mathbb{E}}
\nc{\BQ}{\mathbb{Q}}

\nc{\BH}{{\mathbb H}}
\nc{\BS}{{\mathbb S}}

\nc{\nl}{&\hphantom{=}}

\nc{\diff}{d}

\numberwithin{equation}{section}

\usepackage{soul}
\setstcolor{red}

\usepackage[author={MAK}]{pdfcomment}

\begin{document} 

\renewcommand{\thefootnote}{\fnsymbol{footnote}}
\author{Michael Andreas Klatt\footnotemark[1] \footnotemark[2] , Christian Bair\footnotemark[2] , Hartmut L\"owen\footnotemark[2] ,\\ Ren\'e Wittmann\footnotemark[2] \footnotemark[3]}
\footnotetext[1]{Deutsches Zentrum f\"ur Luft‐ und Raumfahrt (DLR), Institut f\"ur KI Sicherheit, Wilhelm-Runge-Stra\ss e 10, 89081 Ulm, Germany}
\footnotetext[1]{Deutsches Zentrum f\"ur Luft- und Raumfahrt (DLR), Institut f\"ur Materialphysik im Weltraum, 51170 K\"oln, Germany}
\footnotetext[2]{Institut f\"ur Theoretische Physik II: Weiche Materie, 
Heinrich-Heine-Universit\"at D\"usseldorf, 40225 D\"usseldorf, Germany}
\footnotetext[3]{Institut f\"ur Sicherheit und Qualit\"at bei Fleisch, Max Rubner-Institut, D-95326 Kulmbach, Germany}

\title{\Large
  Foundation of classical dynamical density functional theory:
  uniqueness of time-dependent density--potential mappings}
\date{\today}
\maketitle

\begin{abstract} \noindent %
When can we map a classical density profile to an external potential? In
equilibrium, without time dependence, the one-body density is known to
specify the external potential that is applied to the many-body system.
This mapping from a density to the potential is the cornerstone of
classical density functional theory (DFT). Here, we consider
non-equilibrium, time-dependent many-body systems that evolve from a
given initial condition. We derive explicit conditions, for example, no
flux at the boundary, that ensure that the mapping from the density to a
time-dependent external potential is unique. We thus prove the
underlying assertion of dynamical density functional theory
(DDFT)~---~without resorting to the so-called adiabatic approximation
often used in applications. By ascertaining uniqueness for all $n$-body
densities, we ensure that the proof~---~and the physical conclusions
drawn from it~---~hold for general \textit{superadiabatic} dynamics of
interacting systems even in the presence of (known) non-conservative
forces.
\end{abstract}

\noindent

\section{Introduction}

The foundation of classical density functional theory
(DFT)~\cite{mermin_thermal_1965, evans_nature_1979,
evans_density_1992,wu_density_2006} rests on the fact that the one-body
density determines the external potential and hence the underlying
Hamiltonian if the interaction potential is known. In essentially all
relevant cases, there exists a unique mapping from the one-body density
$\rho(x)$  to an external potential $V(x)$ for $x\in\Omega$ in
$d$-dimensional Euclidean space $\Omega\subset\R^d$ and for a given
interaction potential, temperature, and number of particles (or chemical
potential). Remarkably, because of this unique mapping, the one-body
density fully specifies a many-body system in equilibrium if the
interparticle interactions are known. The existence of such a unique
density--potential mapping was first proven in the context of quantum
mechanics by Hohenberg and Kohn~\cite{hohenberg_inhomogeneous_1964},
Kohn and Sham~\cite{kohn_self-consistent_1965}, and
Mermin~\cite{mermin_thermal_1965}. Mermin's generalized arguments can be
directly applied to classical many-body systems as elaborated by
Evans~\cite{evans_nature_1979} and later rigorously confirmed by Chayes,
Chayes, and Lieb~\cite{chayes_inverse_1984}. The unique mapping exists
under mild and natural conditions on the density and interparticle
interactions that essentially assume finite energies. Among others, this
result implies a formal equivalence of Mermin-Evans DFT to the
alternative framework~\cite{dwandaru_variational_2011} based on Levy's
constrained search~\cite{levy_universal_1979} (which does not a priori
restrict to density profiles that are realizable by an external
potential).

Here, we are interested in the generalization of these density--potential
mappings to the time-dependent case, in which all functions additionally
depend on time $t\in[0,\infty)$. If such a mapping exists, it is the
foundation of classical dynamical density functional theory
(DDFT)~\cite{marconi_dynamic_1999, archer_dynamical_2004,
te_vrugt_classical_2020}, as first derived by Marconi and
Tarazona~\cite{marconi_dynamic_1999} from the stochastic Langevin
equation and later by Archer and Evans~\cite{archer_dynamical_2004} from
the corresponding Smoluchowski equation or by Espa\~{n}ol and L\"owen
using the projection operator formalism~\cite{espanol_derivation_2009}.
More specifically, we study the fundamental relation between the
time-dependent one-body density $\rho(x,t)$ and external potential
$V(x,t)$. These considerations also elucidate the role of the one-body
current $j(x,t)$. 

The central question of this paper is: under which conditions can we map
a classical time-dependent density $\rho: \Omega\times[0,\infty)\to \R$
to an external potential $V: \Omega\times[0,\infty)\to \R$? 
Such a mapping always refers to the functions and not their
instantaneous value at time $t$, and it presupposes some well-specified
initial conditions.
Importantly, in this context of statistical physics, a \textit{unique}
mapping from $\rho$ to $V$ implies that the density $\rho$ can only be
realized by the external potential $V$ and not by another potential
$V'\neq V$. However, the term \textit{unique} does not imply that the
mapping is injective. In fact, already in equilibrium, at phase
coexistence, two different densities $\rho \neq \rho'$ can be realized
by the same potential $V$.

As in the case of equilibrium DFT, once we have proven the existence of
a unique density--potential mapping, we can assert that the full
knowledge of the density profile $\rho: \Omega\times[0,\infty)\to \R$ at
all times specifies the forces driving the dynamics at all times and
hence contains all relevant information about the system, including
correlations of any order. Hence, our central question is of both
fundamental importance and practical relevance to the study of
time-dependent many-body systems. 

We, therefore, consider the exact dynamics driven by the full
non-equilibrium interaction force, i.e., we do not rely on an
``adiabatic'' approximation that equates equilibrium and non-equilibrium
correlations (which is usually required for explicit calculations in
DDFT)~\cite{marconi_dynamic_1999, archer_dynamical_2004}. Consequently,
our results also pertain to the recently developed
superadiabatic-DDFT~\cite{tschopp_first-principles_2022,
tschopp_superadiabatic_2023} as well as to the framework of power
functional theory (PFT)~\cite{schmidt_power_2013,schmidt_power_2022}.
Both approaches incorporate ``superadiabatic'' forces
\cite{fortini2014superadiabatic} that are neglected in standard
(adiabatic) DDFT approximations. The underlying variational principle of
PFT, based on Levy's constrained search~\cite{levy_universal_1979},
entails the existence of a unique mapping from the two functions
$\rho(x,t)$ and $j(x,t)$ to the function $V(x,t)$. A recent
reformulation of the PFT formalism suggests that the only relevant field
is $j(x,t)$~\cite{lutsko_oettel_2021}. Since our pursuit of unique
density--potential mappings neither requires an approximation nor a
specific framework, our results shed light on the relation between DDFT
and PFT on a formal level and help, in particular, to better understand
the role of the current.

In quantum mechanics, an argument for the unique mapping from
time-dependent densities $\rho(x,t)$ to potentials $V(x,t)$ was provided
by Runge and Gross in 1984~\cite{runge_density-functional_1984}, which
became the foundation of time-dependent density functional theory
(TDDFT). Assuming time-analytic potentials and smooth densities, they
linked the quest for uniqueness of the density--potential mapping to
that of the solution for an elliptic partial differential
equation~(PDE). However, as pointed out
later~\cite{xu_current-density-functional_1985,
dhara_density-functional_1987, gross_time-dependent_1990,
gross_introduction_2012}, this solution is unique only under certain
conditions on both $\rho(x,t)$ and $V(x,t)$.
These joint assumptions on the density and potential are more complex
than in equilibrium, where the conditions for a unique $V(x)$ only depend on
$\rho(x)$~\cite{chayes_inverse_1984}. Intuitively speaking, these more
intricate assumptions arise in the time-dependent case because more
states are allowed than in equilibrium. 

For classical systems, Chan and Finken~\cite{chan_time-dependent_2005}
asserted uniqueness following the idea of Runge and
Gross~\cite{runge_density-functional_1984}. However, because higher-body
correlations due to interparticle interactions were omitted, the
argument so far holds only under the adiabatic approximation.
Moreover, no conditions have hitherto been stated for the
uniqueness of classical density--potential mappings. In fact, this
omission is more critical in the classical setting than it would be in
the quantum case since, for the latter, loopholes to unique
mappings are considered to be ``largely
unphysical''~\cite{gross_introduction_2012} and are hence often
neglected. By contrast, diverging potentials are not only relevant but
even common in classical statistical physics.

In this work, we begin to close these two gaps by proving explicit
conditions for the uniqueness of classical density--potential mappings
based on an exact hierarchy for the $n$-body densities. Thus, our
conditions are independent of the adiabatic approximation, and, somewhat
surprisingly, they also allow for (known) non-conservative forces. Thus,
we begin to provide a rigorous foundation of classical DDFT. Since we
follow the proof of Runge and Gross
\cite{runge_density-functional_1984}, we have to assume time-analytic
potentials and smooth densities. We thus provide only sufficient
conditions for uniqueness. In the outlook, we briefly discuss the
general case as an open problem for further research in mathematical
physics.

To this end, we first specify our setting in Sec.~\ref{sec:setting} and
derive the hierarchy of reduced Smoluchowski equations for all
time-dependent $n$-body densities $\rho_n(x_1,\dots,x_n,t)$ in
Sec.~\ref{sec:Smoluchowski}; see Theorem~\ref{thm:smoluchowski}. Since
we are concerned with possibly diverging potentials, we accurately
derive the boundary contributions and find that all corresponding terms
vanish if and only if the Yvon-Born-Green (YBG)-hierarchy holds on
average at the boundary (for which we also provide a physical intuitive
condition).

Then, we prove our main results in Sec.~\ref{sec:uniqueness}, i.e., we
derive generic conditions that guarantee that two given external
potentials $V(x,t)$ and $V'(x,t)$ do not yield the same time-dependent
one-body density $\rho(x,t)$. We use this general result to specifically
show that the density-potential mapping is unique for no-flux boundary
conditions.
More specifically, similar to the idea of Runge and Gross (or Chan and
Finken)~\cite{runge_density-functional_1984, chan_time-dependent_2005},
we assume analytic potentials and can, therefore, reduce the uniqueness
of the mapping to the uniqueness of a solution to a (semi-)elliptic PDE.
In contradistinction to the available proofs in quantum
mechanics~\cite{runge_density-functional_1984,
gross_time-dependent_1990, gross_introduction_2012,
ruggenthaler_existence_2015}, we here use the hierarchy of reduced
Smoluchowski equations. Our proof of Theorem~\ref{thm:surface},
therefore, requires no approximation of $n$-body correlations; in
particular, our uniqueness theorems are not restricted to adiabatic
DDFT.
Moreover, our rephrasing of the problem allows us to obtain a physically
intuitive condition for uniqueness. Theorem~\ref{thm:no-normal-flux}
asserts that a unique solution can be guaranteed for no-flux boundary
conditions or, in fact, any specified flux in or out of the system.

In Sec.~\ref{sec:examples}, we demonstrate that such a simultaneous
condition on the density and potential is inevitable. More specifically,
we present explicit loopholes to uniqueness where the same $\rho(x,t)$
is attained for two different external potentials. Obviously, these
examples violate the conditions of our theorems. For an exponentially
fast decaying density profile, a non-unique external potential must
necessarily include an exponential divergence (in space). In contrast,
if the density profile has heavy tails, already a polynomial divergence
of $V(x,t)$ can lead to non-unique mappings. Hence, the conditions on
the asymptotic behavior have to depend on both $\rho(x,t)$ and $V(x,t)$.

To conclude the discussion of loopholes in Sec.~\ref{sec:examples}, we
embed our findings in the framework of PFT. A unique mapping to an
external potential implies a unique current $j(x,t)$. In contrast, if a
suitable external potential $V'(x,t)$ that violates our conditions is
added, it causes a divergence-free current $j'(x,t)$ that does not
change the density $\rho(x,t)$. Such loopholes have been simulated via a
numerical procedure known as \textit{custom
flow}~\cite{de_las_heras_custom_2019, de_las_heras_perspective_2023}
that determines, in line with PFT, the unique external force field as a
functional of $\rho(x,t)$ and $j(x,t)$. The hierarchy of Smoluchowski
equations from Theorem~\ref{thm:smoluchowski} emphasizes the necessity
of this approach for interacting systems. For the ideal gas (or under
the adiabatic approximation), our analytic formula~\eqref{eq:algo} can
be applied for a systematic construction of non-unique
density--potential mappings for effectively one-dimensional systems. 

Our mathematical findings are summarized and interpreted in physical
terms in Section~\ref{sec:outlookPHYS}, where we focus on the insights
in view of necessary approximations required for explicit predictive
calculations using DDFT and related theories. Finally, in
Section~\ref{sec:outlook}, we conclude and discuss open mathematical
questions. We also offer pathways towards further uniqueness theorems
for a broad class of soft matter systems.

\section{Densities and Smoluchowski operators}
\label{sec:setting}

We here consider an overdamped many-body system with a fixed number of
indistinguishable particles $N>0$ in an open, bounded domain
$\Omega\subseteq\R^d$, $d\geq1$. Furthermore, we assume that $\Omega$ is
connected (otherwise we would consider two or more independent systems)
and that $\Omega$ has a smooth boundary $\partial\Omega$.

We characterize our system by its time-dependent symmetric $N$-body probability 
density 
\begin{align}
  P_{N}:\Omega^N\times [0,\infty) \to [0,\infty); \quad (x^N,t)\mapsto P_N(x^N,t) .
  \label{eq:P_N}
\end{align}
In the following, 
we will write $x^N$ as a shorthand notation for a collection of 
positions $x_1,\ldots,x_N\in\Omega$ and time is denoted by $t\geq 0$.
Note that $P_{N}$ is a simple function of time $t$ but a probability 
density in the spatial coordinates $x^N$, i.e., $P_N(x^N, t)$ is the probability 
density of finding particles at positions $x^N$ at time $t$. 
Hence, at any time $t\geq 0$, the integral over $\Omega^N$ is normalized, 
\begin{align}
  \int P_N(x^N,t)\,dx^N = 1.
  \label{eq:norm}
\end{align}
Here and in the following, each unspecified integral is over the full 
domain.

We obtain the (reduced) $n$-body densities $\rho_n:\Omega^n\times
[0,\infty) \to [0,\infty)$ with $n\leq N$ from the symmetric $N$-body
probability density $P_{N}$ by applying the $n$-body density operator:
\begin{align}
  \begin{aligned}
  \rho_n(x^n,t) 
  :=\;&\int P_N(y^N,t)\sum_{i_1\neq i_2 \neq \dots\neq i_n}\delta(x_{1}-y_{i_1})\dots\delta(x_{n}-y_{i_n})\,dy^N\\
   =\;&\frac{N!}{(N-n)!}\int P_N(y^N,t)\delta(x_1-y_1)\dots\delta(x_n-y_n)\,dy^N\\
   =\;&\frac{N!}{(N-n)!}\int P_N(x^n,y^{N-n},t)\,dy^{N-n},
  \end{aligned}
  \label{eq:rho_n}
\end{align}
where $\delta$ denotes a Dirac delta distribution. In the following we use 
the shorthand notation for the number density $\rho:=\rho_1$.

The interaction between the particles is given by a two-body force
$K_2:\Omega^2\to\R^d$, which consists of a pair potential
$U:\Omega^2\to\R$ and a divergent-free non-conservative force
$Q:\Omega^2\to\R^d$, i.e., $K_2(x,y) := - \nabla_x U(x,y) + Q(x,y)$,
where $Q(x,y)$ and $K_2(x,y)$ denote forces of a particle at $y$ acting
on a particle at $x$. As usual, the pair potential is symmetric and only
depends on the relative distance, i.e., we define
$U(x-y):=U(x-y,0)=U(x,y)$ for $x,y\in\Omega$. We assume that $K_2\in
C^1(\Omega^2)$ and, that $K_2$ is bounded. The inverse temperature
$\beta$ and diffusion constant $D$ are fixed.

The many-body system is also subject to an external potential
$V:\Omega\times[0,\infty)\to \R$ that is twice continuously differentiable
in space (i.e., for $t\geq 0$, the function $V_t:\Omega\to\R, x\mapsto
V(x,t)$ is in $C^2(\Omega)$) and $\nabla_x V_{t}:\Omega\to \R^d,$
$x\mapsto\nabla_x V(x,t)$ is continuously extensible to $\partial\Omega$.
(Note, however, that our bounded system can be physically interpreted as
being confined within hard walls, i.e., an infinitely steep potential
on $\partial\Omega$ in addition to $V$.)

Furthermore we consider a divergent-free, non-conservative, one-body,
time-dependent force $R:\Omega\times[0,\infty)\to\R^d$. We again assume,
that $R$ is infinitely differentiable in time and continuous
differentiable in space. We combine these external one-body forces in
the definition of $K_1(x,t) := - \nabla_x V(x,t) + R(x,t)$.

The evolution of $P_{N}$ under such an external potential obeys 
the following $N$-body Smoluchowski equation:
\begin{align}
   {\partial_t P_N(x^N,t)} &= 
     D \sum_{i=1}^{N} \nabla_{x_i}^2P_N(x^N,t)
     - D\beta \sum_{i=1}^{N}\nabla_{x_i}\left[P_N(x^N,t)F_{N,i}(x^N,t)\right],
  \label{eq:smoluN}
\end{align}
where the force $F_{N,i}:\Omega^N\times [0,\infty)\to \R^d$ on particle $i\in\{1,2,\dots, N\}$ is 
defined as
\begin{align}
  F_{N,i}(x^N,t):= K_1(x_i,t) + \sum\limits_{\substack{j=1 \\ j\neq i}}^N K_2(x_i,x_j) .
  \label{eq:force}
\end{align}
The same definition of $F_{N,i}$ holds for any number of particles,
say $n< N$.
Let us also point out here that the index $i$ always refers to the
$i$th argument, e.g.,
$F_{n+1,n+1}(x^n,y,t)=K_1(y,t)+\sum_{j=1}^nK_2(y,x_j)$.

In shorthand notation, we combine all forces into a single vector 
$F_N:\Omega^N\times[0,\infty)\to\R^{dN}$ (and analogously define the gradient $\nabla_{x^N}$).
We, moreover, define the $N$-body current field $J_N:\Omega^N\times[0,\infty)\to\R^{dN}$ as
\begin{align}
  J_N(x^N,t):= -D \nabla_{x^N} P_N(x^N,t) + D\beta  P_N(x^N,t) F_N(x^N,t)
  \label{eq:current_defined}
\end{align}
that obeys the continuity equation
\begin{align}
  {\partial_t P_N(x^N,t)} = - \nabla_{x^N} J_N(x^N,t),
  \label{eq:current_smoluN}
\end{align}
which is then equivalent to the Smoluchowski equation.
By defining the Smoluchowski operator
\begin{align}
  \hat{\mathcal{O}}_N := D \sum_{i=1}^{N}\nabla_{x_i} \left[ \nabla_{x_i} - \beta F_{N,i}(x^N,t) \right],
\end{align}
we can write the Smoluchowski equation~\eqref{eq:smoluN} more succinctly as
\begin{align}
  {\partial_t P_N(x^N,t)} &= \hat{\mathcal{O}}_N P_N(x^N,t).
  \label{eq:smoluoperator}
\end{align}
For our derivation of a reduced Smoluchowski equation for 
$\rho_n(x^n,t)$ in the next section, it is useful to define 
\textit{partial} Smoluchowski operators as
\begin{align}
  \hat{\mathcal{O}}_{n,N}^{-} & := D \sum_{i=1}^n   \nabla_{x_i} \left\{ \nabla_{x_i} - \beta F_{N,i}(x^N,t) \right\},\\
  \hat{\mathcal{O}}_{n,N}^{+} & := D \sum_{i=n+1}^N \nabla_{x_i} \left\{ \nabla_{x_i} - \beta F_{N,i}(x^N,t) \right\}.
\end{align}

The behavior of the system is determined by an initial value boundary
problem that, in our case, is defined by the Smoluchowski
equation~\eqref{eq:smoluoperator}, the initial condition
$P_N^{(0)}:\Omega^N\to \R, x^N\mapsto P_N^{(0)}(x^N):= P_{N}(x^N,0)$ at
time $t=0$ (for all spatial coordinates), and a boundary condition on
$\partial \Omega$ (for all times).
Since we here consider a system with a fixed number of particles,
a natural choice is that there is no $N$-body flux in or out of the
system, i.e.,
\begin{align}\label{eq:no-flux}
  n(x^N) J_N(x^N,t) = 0, \text{ for all } x^N\in\partial\Omega^N,
\end{align}
where $n(x^N)$ denotes the outward unit normal on $\partial\Omega^N$.
Another prominent example (especially for simulations) are periodic
boundary conditions. In both examples, the number of particles is
conserved.

We say that a solution $P_{N}$ is \textit{well behaved} if it has
the following properties: 
\begin{enumerate}[label=($\roman*$)]
  \item $P_N\in C^2(\Omega^N\times [0,\infty))$, i.e., for $t=0$ twice
  differentiable from the right-hand side;
  \item $P_N$ and $\nabla_{x^N} P_N$ are continuously extensible to
  $\partial\Omega^N$; 
  \item $\partial_t P_{N}$, $\hat{\mathcal{O}}_{n}P_{N}$,
  $\hat{\mathcal{O}}_{n,N}^{\pm}P_{N}$, $\nabla_x K_2(x,y)$ and $\nabla_{x^N}
  P_N$ are bounded for all $n=1,2,\dots N$; 
  \item the average $n$-body interaction force $E_{n,i}: \Omega^n\times
  [0,\infty) \to \R$ on particle $i\in\{1,\ldots,n\}$ that is defined by 
  \begin{align}
    E_{n,i}(x^n,t):=  \int \rho_{n+1}(x^n,y,t) K_2(x_i,y) \,dy
    \label{eq:Ei}
  \end{align}
  exists and is continuously differentiable on $\Omega^n$ and for $t\geq 0$. 
  For convenience, we also define $E_{n,i}\equiv 0$ for $n\geq N$. The index 
  is analogously defined to that of the force $F_{n,i}$.
\end{enumerate}

Here, we always assume the existence of a well-behaved solution. Even
though a proof of existence is beyond the scope of this paper, we
briefly discuss in the outlook a possible approach and conditions that
can be expected.

The main question of this paper can now be summarized as follows.
Consider $N$ particles with known interparticle interaction forces.
Given the initial conditions $P_N^{(0)}$, does the density $\rho$ as a
function of time and space uniquely specify the potential $V$, or in
more mathematical terms, is a mapping from $\rho$ to $V$ well defined?

As noted in the introduction, we ask for a mapping from $\rho$ to a
single $V$, but this mapping will in general not be injective, i.e., two
densities can be mapped to the same potential. For example, already in
equilibrium, at phase coexistence, different density profiles can both
be realized by the same potential. 


\section{Hierarchy of reduced Smoluchowski equations}
\label{sec:Smoluchowski}

In the physics literature, boundary terms are often neglected in the
derivation of a reduced Smoluchowski equation~\cite{lowen_dynamical_2017}.
These boundary terms may be essential, however, to derive necessary
conditions for a unique density--potential mapping (since non-unique
loopholes involve diverging external potentials). 

We, therefore, first derive a reduced Smoluchowski equation paying
special attention to the boundary terms. Moreover, since we do not rely
on the adiabatic approximation but instead consider the exact
dependencies between $n$-body densities, we derive a complete set of
reduced Smoluchowski equations for all orders.

\begin{theorem}
  \label{thm:smoluchowski}
  The reduced $n$-body density $\rho_n$ with $1\leq n < N$
  obeys the following reduced Smoluchowski equation:
\begin{align}
  \begin{aligned}
    {\partial_t \rho_n(x^n,t)}
    &= D\sum_{i=1}^{n}\nabla_{x_i}\left\{ \left[ \nabla_{x_i} - \beta F_{n,i}(x^n,t) \right] \rho_n(x^n,t)- \beta  E_{n,i}(x^n,t)\right\} \\
  &\,\, + D\oint_{\partial\Omega}\left\{ 
	\left[\nabla_{y} -\beta F_{n+1,n+1}(x^n,y,t)\right]\rho_{n+1}(x^n,y,t) -{\beta } E_{n+1,n+1}(x^n,y,t)
    \right\}\,dS(y),
  \end{aligned}
  \label{eq:reducedSmoluchowski}
\end{align}
  where the last term represents a surface integral with respect to the outward unit normal vector.
\end{theorem}
\begin{proof}
Under our assumptions, we can apply the $n$-body density operator 
to the $N$-body Smoluchowski differential 
equation~\eqref{eq:smoluoperator}. First, we use
\begin{align}
\partial_t \rho_n(x^n,t) = \frac{N!}{(N-n)!}\int \partial_t P_N(x^n,y^{N-n},t)\,dy^{N-n}  
\end{align}
and $\hat{\mathcal{O}}_N = \hat{\mathcal{O}}_{n,N}^{-} +\hat{\mathcal{O}}_{n,N}^{+}$
to obtain
\begin{align}
  \begin{aligned}
  {\partial_t \rho_n(x^n,t)}
  &=    \frac{N!}{(N-n)!}\int \hat{\mathcal{O}}_{n,N}^{-}P_N(x^n,y^{N-n},t)\,dy^{N-n}\\
  \nl + \frac{N!}{(N-n)!}\int \hat{\mathcal{O}}_{n,N}^{+}P_N(x^n,y^{N-n},t)\,dy^{N-n}.
  \end{aligned}
  \label{eq:partialsplit}
\end{align}
To simplify the first term on the right-hand side, we note that 
\begin{align}
  \hat{\mathcal{O}}_{n,N}^{-} = \hat{\mathcal{O}}_{n} - D\beta\sum_{i=1}^{n}\nabla_{x_i}\left[\sum_{j=1}^{N-n} K_2(x_i,y_j) \right]
\end{align}
and hence
\begin{align}
\begin{aligned}
  &\frac{N!}{(N-n)!}\int \hat{\mathcal{O}}_{n,N}^{-} P_N(x^n,y^{N-n},t)\,dy^{N-n}
  = \hat{\mathcal{O}}_{n} \rho_n(x^n,t)\\
  &\quad- D\beta\frac{N!}{(N-n)!}\int \sum_{i=1}^{n}\nabla_{x_i}\left[ P_N(x^n,y^{N-n},t)\sum_{j=1}^{N-n} K_2(x_i,y_j) \right]\,dy^{N-n}.
\end{aligned}
\label{eq:tmp}
\end{align}
Using the average $n$-body interaction force $E_{n,i}$ from~\eqref{eq:Ei},
we can further simplify the remaining integral: 
\begin{align}
  &\frac{N!}{(N-n)!}\int \sum_{i=1}^{n}\nabla_{x_i}\left[ P_N(x^n,y^{N-n},t)\sum_{j=1}^{N-n} K_2(x_i,y_j) \right]\,dy^{N-n} \\
  &\qquad = \sum_{i=1}^{n}\nabla_{x_i}\int \sum_{j=1}^{N-n}\frac{N!}{(N-n)!}  P_N(x^n,y^{N-n},t) K_2(x_i,y_j) \,dy^{N-n} \allowdisplaybreaks\\
  &\qquad = \sum_{i=1}^{n}\nabla_{x_i}\int \frac{N!}{[ N-(n+1) ]!} \int P_N(x^n,y,z^{N-(n+1)},t) \, dz^{N-(n+1)} K_2(x_i,y) \, dy \\ 
  &\qquad = \sum_{i=1}^n \nabla_{x_i} E_{n,i}(x^n,t).
\end{align}
Inserting this result in \eqref{eq:tmp} and finally in \eqref{eq:partialsplit}, we have
\begin{align}
  \begin{aligned}
  {\partial_t \rho_n(x^n,t)}
  &= \hat{\mathcal{O}}_{n} \rho_n(x^n,t)- D\beta \sum_{i=1}^n \nabla_{x_i} E_{n,i}(x^n,t) + B(x^n,t),
  \end{aligned}
  \label{eq:withboundary}
\end{align}
where we define the \textit{boundary term} by
\begin{align}
   B(x^n,t):=\; &
   \frac{N!}{(N-n)!}\int \hat{\mathcal{O}}_{n,N}^{+} P_N(x^n,y^{N-n},t)\,dy^{N-n}\\
   =\;& D \frac{N!}{(N-n)!}\int  \sum_{i=1}^{N-n} \nabla_{y_i} \left\{ \nabla_{y_i} - \beta F_{N,n+i}(x^n,y^{N-n},t) \right\} P_N(x^n,y^{N-n},t)\,dy^{N-n}\\
   =\;& D \frac{N!}{[N-(n+1)]!}\oint_{\partial\Omega} \int \left\{ \nabla_{y} - \beta F_{N,n+1}(x^n,y,z^{N-(n+1)},t) \right\} \\
   &\hspace{6.3cm}\times P_N(x^n,y,z^{N-(n+1)},t)\,dz^{N-(n+1)}\,dS(y).
\end{align}
The last equality holds by the divergence theorem.
To prove~\eqref{eq:reducedSmoluchowski}, it remains to show that
\begin{align}
  \frac{N!}{[N-(n+1)]!} \int &\left\{ \nabla_{y} - \beta F_{N,n+1}(x^n,y,z^{N-(n+1)},t) \right\} P_N(x^n,y,z^{N-(n+1)},t)\,dz^{N-(n+1)}\\
   &=\left[\nabla_{y} -\beta F_{n+1,n+1}(x^n,y,t)\right]\rho_{n+1}(x^n,y,t) -{\beta } E_{n+1,n+1}(x^n,y,t).
\end{align}
This assertion follows from the fact that
\begin{align}
  F_{N,n+1}(x^n,y,z^{N-(n+1)},t) = F_{n+1,n+1}(x^{n},y,t) + \sum_{j=1}^{N-(n+1)} K_2(y,z_j) 
\end{align}
and (for $n<N-1$)
\begin{align}
  &\frac{N!}{[N-(n+1)]!} \int \sum_{j=1}^{N-(n+1)} K_2(y,z_j) P_N(x^n,y,z^{N-(n+1)},t)\,dz^{N-(n+1)}\\
  &\hspace{1cm}= \frac{N!}{[N-(n+2)]!} \int K_2(y,z)\int P_N(x^n,y,z,v^{N-(n+2)},t)\, dv^{N-(n+2)}\,dz,
\end{align}
where the last expression is, by definition, equal to $E_{n+1,n+1}(x^n,y,t)$.
\end{proof}

\begin{remark}\label{re:YBG}
  The two expressions in curly brackets in the first and second line of
  the reduced Smoluchowski equation~\eqref{eq:reducedSmoluchowski} for
  $\rho_n$ are those of the YBG
  hierarchy~\cite[Sec.~4.2]{hansen_theory_2013} for order $n$ and $n+1$,
  respectively. In equilibrium, the two expressions always vanish, which
  is in agreement with $\partial_t \rho_n(x^n,t)= 0$ for all
  $x^n\in\Omega^n$ and $t\geq 0$. Out of equilibrium, these boundary
  terms vanish for all orders if and only if the YBG hierarchy holds
  \textit{on average} at the boundary.
\end{remark}

In the following, we denote the components of $J_N$ that correspond to
the $i$th particle by $J_{N,i}:\Omega^N\times[0,\infty)\to\R^{d}$ (with
$i=1,\dots,N$). We are interested in physically relevant boundary
problems for which the following condition holds
\begin{align}
  \int \nabla_{x_i} J_{N,i}(x^N,t) \,dx_i = \oint_{\partial\Omega} J_{N,i}(x^{N},t) \,dS(x_i) & = 0, \text{ for all }i=1, \dots N.
  \label{eq:boundary}
\end{align}
This condition guarantees that a time-dependent version of the YGB
hierarchy holds on average at the boundary for all orders.

\begin{proposition}
  \label{prop:YBG}
  If \eqref{eq:boundary} holds, then the boundary term in
  \eqref{eq:reducedSmoluchowski} vanishes for all $n < N$.
\end{proposition}
\begin{proof}
Combining $F_{n,1},\dots F_{n,n}$ into one vector-valued function $F_n:
\Omega^n\times [0,\infty)\to \R^{dn}$ (analogous to $F_N$) and
$E_{n,1},\dots E_{n,n}$ into $E_n: \Omega^n\times [0,\infty)\to
\R^{dn}$, we define the $n$-body current $J_n: \Omega^n\times
[0,\infty)\to \R^{dn}$ as
\begin{align}
    J_n(x^n,t) := -D\nabla_{x^n} \rho_n(x^n,t) + \beta D F_n(x^n,t) \rho_n(x^n,t) +\beta D E_n(x^n,t).
    \label{eq:def-n-current}
\end{align}
We use the similar notation as by $J_{N,i}$. Then
\eqref{eq:reducedSmoluchowski} can then be rewritten as
\begin{align}
    {\partial_t \rho_n(x^n,t)}
    &= - \nabla_{x^n} J_n(x^n,t) - \oint_{\partial\Omega} J_{n+1,n+1}(x^n,y,t) \,dS(y).
  \label{eq:reducedSmoluchowski-n-current}
\end{align}
Analogous to the proof of Theorem~\ref{thm:smoluchowski}, we obtain for
$i=1,\dots,n$
\begin{align}
  \int J_{N,i}(x^{n},y^{N-n},t)\,dy^{N-n} & = \frac{(N-n)!}{N!} J_{n,i}(x^{n},t) 
\end{align}
and thus obtain for the boundary term in
\eqref{eq:reducedSmoluchowski-n-current}:
\begin{align}
  \oint_{\partial\Omega} J_{n+1,n+1}(x^{n},y,t) \,dS(y) 
  & = \frac{N!}{(N-n)!} \int \oint_{\partial\Omega} J_{N,n+1}(x^{n},y,z^{N-n-1},t) \,dS(y) dz^{N-n-1} =0,
\end{align}
where the last equation holds due to \eqref{eq:boundary}. 
\end{proof}

In the following, we always assume that \eqref{eq:boundary} holds. For
example, this condition is guaranteed both by \eqref{eq:no-flux} or
periodic boundary conditions mentioned above. Thus, we recover the
well-known (reduced) Smoluchowski equation for the one-body density:
\begin{align}
  \begin{aligned}
    \partial_t \rho(x,t)
    &= D\nabla_{x}\left\{ \left[ \nabla_{x} + \beta \nabla_x V(x,t) - \beta R(x,t) \right] \rho(x,t) - \beta \int \rho_{2}(x,y,t) K_2(x,y)\, dy\right\} .
  \end{aligned}
  \label{eq:rho1}
\end{align}
More generally, for the $n$-body densities, we obtain:
\begin{align}
  \begin{aligned}
    {\partial_t \rho_n(x^n,t)}
    &= D\sum_{i=1}^{n}\nabla_{x_i}\bigg\{ \left[ \nabla_{x_i} + \beta \nabla_{x_i} V(x_i,t) - \beta R(x_i,t)\right] \rho_n(x^n,t) \\
    &\qquad - \beta \rho_n(x^n,t) \sum\limits_{\substack{j=1 \\ j\neq i}}^n K_2(x_i,x_j) - \beta \int \rho_{n+1}(x^n,y,t)K_2(x_i,y) \, dy \bigg\}.
  \end{aligned}
  \label{eq:rhon}
\end{align}
Based on \eqref{eq:rho1}, we
define the one-body current $j:\Omega\times [0,\infty) \to \R^d$ as
\begin{align}
  \begin{aligned}
    j(x,t) &:= -D\left[ \nabla_{x} + \beta \nabla_x V(x,t) - \beta R(x,t) \right] \rho(x,t)  + D\beta \int \rho_{2}(x,y,t)K_2(x,y)\, dy,
    \label{eq:onebodycurrent}
  \end{aligned}
\end{align}
so that it obeys the following continuity equation
\begin{align}
  \partial_t \rho(x,t) &= -\nabla_{x}j(x,t).
  \label{eq:continuityRho1}
\end{align}
Our definition \eqref{eq:onebodycurrent} is equivalent to the ensemble
average of a current operator~\cite{schmidt_power_2022}.

We define the normal one-body current $j_{\perp}:\partial\Omega\times
[0,\infty) \to \R$ at the boundary via an extension of $j$ from
\eqref{eq:onebodycurrent} to $\partial\Omega$:
\begin{align}
  \begin{aligned}
  j_{\perp}(x,t) := n(x) j(x,t) 
  =\;& -Dn(x) \left[ \nabla_{x} + \beta \nabla_x V(x,t) - \beta R(x,t) \right] \rho(x,t) \\
  &+ D\beta n(x)\int \rho_{2}(x,y,t)K_2(x,y)\, dy,
  \end{aligned}
  \label{eq:normalflux}
\end{align}
where $n$ denotes the outward unit normal on $\partial \Omega$. By
property ($iii$) of the well behaved solution and by the type of
potential we consider, $j_{\perp}$ is well defined. As before, the
product of vectors is consistently interpreted as a scalar product.

Boundary condition \eqref{eq:no-flux} of a vanishing normal $N$-body flux
also implies a corresponding condition on the one-body current:
\begin{align}
    j_{\perp}(x,t) =0 \quad \text{for all } x\in\partial\Omega, t\geq 0.
    \label{eq:no-one-body-flux}
\end{align}
Intuitively speaking, the normal component of the one-body current field
vanishes at the boundary $\partial\Omega$. It is a common boundary
condition in physics. In the following, any boundary condition that
implies \eqref{eq:no-one-body-flux} is called a \textit{no-flux boundary
condition}.


\section{Uniqueness theorems}
\label{sec:uniqueness}

We now turn to the central question of this paper. Given an initial 
condition $P_N^{(0)}$, does $\rho$ uniquely specify $V$, i.e., is
a mapping from $\rho$ to $V$ well defined?

There are two obvious limitations to uniqueness. First, the mapping can
only be unique for $x\in\Omega$ where and when $\rho(x,t) > 0$.
Variations in $V$ outside the support of $\rho$, i.e., within the
complement of the set $\textrm{supp} (\rho) :=\{(x,t)\in
\Omega\times[0,\infty):\rho(x,t)>0\}$, do not change the time evolution
of the system as determined by the Smoluchowski equation. Secondly,
adding a time-dependent constant to the potential does not change the
time evolution either. 
We can combine both limitations in a single statement, i.e., if the
difference of the potentials is a time-dependent constant on the support
of $\rho$, then it has no effect on the density. 

\begin{definition}\label{def:diffequiv}
  Two external potentials $V$ and $V'$ are said to be \textit{diffusion
  equivalent}, $V\sim V'$, for a given one-body density $\rho$ if the
  difference $d_V:\Omega\times[0,\infty)\to \R ; (x,t)\mapsto
  d_V(x,t):=V(x,t)-V'(x,t)$ is only a function of time on
  $\mathrm{supp}(\rho)$. 
\end{definition}

This definition allows us to formulate our strategy of proof more
specifically. In the following, we consider two systems with densities
$\rho$ and $\rho'$ and with external potentials $V$ and $V'$. Both
systems start from the same initial condition $P_N^{(0)}$, and they have
the same boundary conditions, pair potentials, and non-conservative
forces.

Our aim is to derive conditions for which an equivalence of $\rho$ and
$\rho'$ implies diffusion equivalence of $V$ and $V'$, or in other words
that the difference
\begin{align}
  d_V(x,t) & := V(x,t) - V'(x,t)
  \label{eq:def_d}
\end{align}
is diffusion equivalent to a function that is identically zero.
For convenience, we will actually show the contrapositive. If the two
potentials are not diffusion equivalent, then the densities must differ,
and thus $\rho$ uniquely determines $V$.

A key step in the proof is to reduce the question of a unique mapping to
that of the unique solvability of a (semi-)elliptic PDE. As discussed in
the introduction, this approach is partly similar to the argument by
Runge and Gross (or Chan and
Finken)~\cite{runge_density-functional_1984, chan_time-dependent_2005},
but it differs in that we use the hierarchy of reduced Smoluchowski
equations from Theorem~\ref{thm:smoluchowski}. We also pay close
attention to a rigorous treatment of the boundary terms. Additionally,
we prove that a no-flux boundary condition always implies uniqueness. 

The main advantage of our strategy is a physically intuitive proof that
helps to clarify the essential physical questions. This intuition comes
at the price of two additional assumptions that could possibly be
avoided  by alternative methods, like a fix-point scheme that has
already been employed in the quantum
case~\cite{ruggenthaler_global_2011, ruggenthaler_existence_2015}. We
summarize these two assumptions via the following definition.

\begin{definition}
  \label{def:analyticity}
  We call a many-body system \textit{analytically accessible} if the
  following two conditions hold:
  \begin{enumerate}
    \item[(A1)] the external potential $V$ is real analytic in time for
      $t\geq 0$ and for all $k\in\N_0$, $\partial^k_t V\in
      C^2(\Omega\times[0,\infty))$ and $\partial^k_t V(\cdot,t=0)$ is
      continuously extensible to $\partial\Omega$;
    \item[(A2)] the $n$-body densities $\rho_n$ for all $n=1,2,\dots N$
      and non-conservative one-body forces $R$ are infinitely often
      differentiable in time from the right at $t= 0$ and for all $k\in
      \N_0$, $\partial_t^k\rho_n\in C^2(\Omega^n\times[0,\infty))$ is
      bounded and $\partial_t^kR\in C^1(\Omega\times[0,\infty))$.
  \end{enumerate}
\end{definition}

More specifically, concerning condition (A1), by including the start
time $t=0$, we assume that the potentials are right differentiable; and
since for each $x\in\Omega$, the functions $V_x:[0,\infty)\to\R,
t\mapsto V(x,t)$ and $V'_x:[0,\infty)\to\R, t\mapsto V'(x,t)$ are real
analytic, the corresponding Taylor series at the origin converge in some
neighborhood of $0$. Hence, the derivatives at the origin uniquely
specify the potential at all times (according to the identity theorem
for analytic functions). With regards to condition (A2), note that we do
not require $\rho_n$ and $\rho_n'$ to be time analytic. 

\begin{proposition}\label{prop:positive_P}
    The $N$-body probability density $P_N(x^N,t) > 0$ for
    $x^N\in\Omega^N$ and $t>0$.
\end{proposition}
\begin{proof}
    Due to \eqref{eq:norm} for $t=0$, $\int
    P_N^{(0)}(x^N)\,dx^N = 1$ and since $P_N^{(0)}(x^N) \geq 0$, the
    support of $P_N^{(0)}$ is not empty and there exist a
    $x_1^N\in\Omega^N$ with $P_N^{(0)}(x_1^N) > 0$. Suppose
    $P_N(x^N_0,t_0) = 0$ for one $x^N_0\in\Omega^N$ and $t_0>0$. Since
    $\Omega$ is open and connected, there is a connected domain $D$
    compactly contained in $\Omega$ ($\bar{D} \subset\Omega$) with
    $x^N_0,x^N_1\in D^N$. From $P_N(x^N_0,t_0) = 0$ follows $\inf_{x^N\in
    D^N} P_N(x^N,t_0) = 0$.

    According to Theorem~10~in~\cite[Chapter~7.1]{evans_partial_2010}, 
    we obtain for $0<t<t_0$ that
    \begin{align}
        \sup_{x^N\in D^N} P_N(x^N,t) \leq C \inf_{x^N\in D^N} P_N(x^N,t_0) = 0
        \label{eq:Harnack}
    \end{align}
    with a constant $C$ that depends on $D^N$, $t$, $t_0$, and the
    coefficients of $\hat{\mathcal{O}}_N$.
    Theorem~10~in~\cite[Chapter~7.1]{evans_partial_2010} assumes
    continuous coefficients of the parabolic differential equation,
    i.e., the assertion only holds for potentials $V\in C^2(\Omega)$, $R
    \in C(\Omega)$, and $K_2 \in C(\Omega^2\times [0,\infty))$, and
    diverging potentials are excluded; all of which is in agreement with
    our assumptions in Sec.~\ref{sec:setting}. 

    By \eqref{eq:Harnack}, $P_N(x^N,t)=0$ for all $x^N\in D^N$ and
    $0<t<t_0$ and specifically $P_N(x_1^N,t)=0$.
    Since $P_N\in C(\Omega^N\times [0,\infty))$, we have
    $P_N(x_1^N,0)=P_N^{(0)}(x_1^N)=0$, which is a contradiction to
    $P_N^{(0)}(x_1^N) > 0$. Therefore, $P_N(x^N,t)>0$ and accordingly
    $\rho(x,t)>0$ for all $x^N\in\Omega^N, x\in\Omega$ and $t>0$.
\end{proof}

To simplify our proofs, we assume in the following that
$P_N^{(0)}(x^N)>0$ holds for all $x^N\in\Omega^N$, which we here denote
by $P_N^{(0)}>0$. With Proposition~\ref{prop:positive_P}, this
assumption is innocuous from a physical point of view.

Since we assume that the potentials are analytic, we will consider the
time derivatives of their difference, i.e., we define
$d_V^{(k)}:\Omega\to\R$ for $k\in\N_0$:
\begin{align}
  x\mapsto d_V^{(k)}(x) := \partial_t^k d_V(x,t)\Bigr|_{t=0}.
\end{align}
By the fact that $V$ and $V'$ are analytic in time and
$\textrm{supp}(\rho) = \Omega\times [0,\infty)$, $V$ and $V'$ are
diffusion equivalent if and only if $\nabla_x d_V(x,t) = 0$ for all
$x\in\Omega$ and $t\geq 0$, which is equivalent to $\nabla_x
d_V^{(k)}(x) = 0$ for all $x\in\Omega$ and $k\in\N_0$.
Let $V$ and $V'$ be not diffusion equivalent; then there exists a
smallest non-negative integer, say $l$, for which $\nabla d_V^{(l)}
\not\equiv 0$.

The proof of our theorems rests on the following lemma. Importantly, it
allows an exact treatment of the average $n$-body interaction forces for
all orders of $n$ (via the hierarchy of reduced Smoluchowski equations).

\begin{lemma}\label{lemma}
  Given two analytically-accessible many-body systems with identical
  initial conditions $P_N^{(0)}>0$ and boundary conditions. Let their
  external potentials $V$ and $V'$ be not diffusion equivalent, and
  denote by $l\in\N_0$ the smallest integer for which $\nabla d_V^{(l)}
  \not\equiv 0$. 
  Then
  \begin{align}
    \partial_t^k \left[\rho_n(x^n,t)-\rho_n'(x^n,t)\right]\Bigr|_{t=0} \equiv 0
    \quad \text{for all } n=1,2,\dots N \text{ and } k=0,1,\dots l
    \label{eq:vanish}
  \end{align}
  and
  \begin{align}
    \partial_t^{l+1} \left[\rho(x,t)-\rho'(x,t)\right]\Bigr|_{t=0}  = D\beta \nabla_x \left[ \rho(x,0) \nabla_x d_V^{(l)}(x) \right].
    \label{eq:toshow}
  \end{align}
\end{lemma}
\begin{proof}
  We first prove \eqref{eq:vanish} by an induction-like argument. This
  equation obviously holds for $k=0$ by \eqref{eq:rho_n} because both
  many-body systems start from the same initial condition $P_N^{(0)}$.
  In the case $l>0$, assume that \eqref{eq:vanish} holds for all
  $k=0,1,\dots m$ for some $m < l$. We must now verify that
  \eqref{eq:vanish} also holds for $k=m+1$. Therefore, we subtract the
  reduced Smoluchowski equations \eqref{eq:rhon} for the $n$-body
  densities of the two many-body systems, take~$m$ additional time
  derivatives and evaluate the derivatives at $t=0$:
  \begin{subequations}
  \begin{align}
      \partial_t^{m+1}& \left[\right.\rho_n(x^n,t)-\left.\rho_n'(x^n,t)\right]\Bigr|_{t=0} 
      = D\sum_{i=1}^{n}\nabla_{x_i}^2 \partial_t^m\left[\rho_n(x^n,t)-\rho_n'(x^n,t)\right]\Bigr|_{t=0} \label{eq:dropout-1} \\
      & + D\beta \sum_{i=1}^{n}\nabla_{x_i} \partial_t^m\left[\rho_n(x^n,t)\nabla_{x_i} V(x_i,t)-\rho_n'(x^n,t)\nabla_{x_i} V'(x_i,t)\right]\Bigr|_{t=0} \label{eq:dropout-2}  \\
      & - D\beta \sum_{i=1}^{n}\nabla_{x_i} \partial_t^m\left[\rho_n(x^n,t)-\rho_n'(x^n,t)\right] R(x_i,t) \Bigr|_{t=0} \label{eq:dropout-3} \\
      & - D\beta \sum_{\substack{i,j=1\\ j\neq i}}^{n}\nabla_{x_i} \left\{\partial_t^m\left[\rho_n(x^n,t)-\rho_n'(x^n,t)\right]\Bigr|_{t=0} K_2(x_i,x_j)\right\} \label{eq:dropout-5} \\
      & - D\beta \sum_{i=1}^{n}\nabla_{x_i} \int \partial_t^m\left[ \rho_{n+1}(x^n,y,t) - \rho_{n+1}'(x^n,y,t)\right]\Bigr|_{t=0}K_2(x_i,y)\, dy . \label{eq:dropout-7}
  \end{align}
  \end{subequations}
  The terms \eqref{eq:dropout-1} on the right-hand side,
  \eqref{eq:dropout-5} and \eqref{eq:dropout-7} vanish directly since
  \eqref{eq:vanish} holds for $k=m$ by our induction hypothesis. For the
  remaining derivative in the second term \eqref{eq:dropout-2}, we have
  \begin{align}
    &\partial_t^m\left[\rho_n(x^n,t)\nabla_{x_i} V(x_i,t)-\rho_n'(x^n,t)\nabla_{x_i} V'(x_i,t)\right]\Bigr|_{t=0} \\
     &\,= \sum_{k=0}^m \binom{m}{k} \left[\partial_t^k\rho_n(x^n,t)\Bigr|_{t=0}\nabla_{x_i} \partial_t^{m-k}V(x_i,t)\Bigr|_{t=0}-\partial_t^k\rho_n'(x^n,t)\Bigr|_{t=0}\nabla_{x_i}\partial_t^{m-k} V'(x_i,t)\Bigr|_{t=0}\right]   \\
     &\,= \sum_{k=0}^m \binom{m}{k} \left[\partial_t^k\rho_n(x^n,t)\Bigr|_{t=0}\nabla_{x_i} d_V^{(m-k)}(x_i)\right],
     \label{eq:product}
  \end{align}
  where the last equality holds again because of our induction
  hypothesis, i.e., we apply \eqref{eq:vanish} for $k\leq m$.
  Now, since $\nabla d_V^{(k)} \equiv 0$ for all $k\leq m<l$,
  assertion \eqref{eq:vanish} follows for all $k=0,1,\dots l$.
  For \eqref{eq:dropout-3}, we consider 
  \begin{align}
      & \partial_t^m\left[\rho_n(x^n,t)-\rho_n'(x^n,t)\right] R(x_i,t) \Bigr|_{t=0} \\
      &\qquad= \sum_{k=0}^m \binom{m}{k} \partial_t^k \left[\rho_n(x^n,t)-\rho_n'(x^n,t)\right] \Bigr|_{t=0} \partial_t^{m-k} R(x_i,t)\Bigr|_{t=0} \stackrel{\eqref{eq:vanish}}{=} 0 . 
  \end{align} 

  To prove \eqref{eq:toshow}, we subtract the reduced Smoluchowski
  equations \eqref{eq:rho1} for the one-body densities of the two
  many-body systems, take $l$ additional time derivatives and evaluate
  the result at $t=0$:
  \begin{align}
    \begin{aligned}
      \partial_t^{l+1} \left[\rho(x,t)-\rho'(x,t)\right]\Bigr|_{t=0} &= 
      D\,\nabla_{x}^2\partial_t^l\left[\rho(x,t)-\rho'(x,t)\right]\Bigr|_{t=0}\\
      &\hphantom{=}\; +D\beta\,\nabla_{x} \partial_t^l\left[\rho(x,t)\nabla_x V(x,t) - \rho'(x,t)\nabla_x V'(x,t) \right]\Bigr|_{t=0}\\
      &\hphantom{=}\; -D\beta\,\nabla_{x} \partial_t^l\left[\rho(x,t) - \rho'(x,t) \right] R(x,t) \Bigr|_{t=0} \\
      &\hphantom{=}\; -D\beta\,\nabla_{x}\int\partial_t^l\left[\rho_{2}(x,y,t)-\rho_{2}'(x,y,t)\right]\Bigr|_{t=0}K_2(x,y)\, dy.
    \end{aligned}
  \end{align}
  The first and last term on the right-hand side vanish by
  \eqref{eq:vanish} and the third term vanishes like
  \eqref{eq:dropout-3}. Thus, we have
  \begin{align}
    \begin{aligned}
      \partial_t^{l+1} &\left[\rho(x,t)-\rho'(x,t)\right]\Bigr|_{t=0} = 
       D\beta\,\nabla_{x} \partial_t^l\left[\rho(x,t)\nabla_x V(x,t) - \rho'(x,t)\nabla_x V'(x,t) \right]\Bigr|_{t=0}\\
       &=D\beta\,\nabla_{x} \sum_{k=0}^l \binom{l}{k} \left[\partial_t^k\rho(x,t)\Bigr|_{t=0}\nabla_x \partial_t^{l-k}V(x,t)\Bigr|_{t=0} - \partial_t^k\rho'(x,t)\Bigr|_{t=0}\nabla_x \partial_t^{l-k}V'(x,t)\Bigr|_{t=0} \right]\\
       &=D\beta\,\nabla_{x} \sum_{k=0}^l \binom{l}{k} \left[\partial_t^k\rho(x,t)\Bigr|_{t=0}\nabla_x d_V^{(l-k)}(x) \right],
    \end{aligned}
  \end{align}
  where the last equality holds again by \eqref{eq:vanish}.
  Since $\nabla d_V^{(l-k)} \equiv 0$ for all $0<k\leq l$, 
  we have proven \eqref{eq:toshow} which concludes the proof.
\end{proof}

The $n$-body densities $\rho_n$ with $n>1$ are highly relevant for the
correct dynamic evolution of the one-body density $\rho$. The preceding
lemma provides control over these contributions since \eqref{eq:vanish}
holds for all $n\geq 1$ and their derivatives up to order $l$. As an
important consequence, \eqref{eq:toshow} and hence our final results on
uniqueness have no explicit condition on $\rho_n$. Our proof via an
induction-like argument reflects the hierarchical structure of the
reduced Smoluchowski equations for $\rho_n$.

Next, we state and prove our first main Theorem~\ref{thm:surface} with a
general and quite formal condition for uniqueness. We then verify this
condition for no-flux boundary conditions in our central
Theorem~\ref{thm:no-normal-flux}. The general condition of
Theorem~\ref{thm:surface} is still useful to easily adapt the proof to
other boundary conditions as in Corollary~\ref{cor:flux}.

\begin{theorem}\label{thm:surface}
  Given two analytically-accessible many-body systems with identical
  initial conditions $P_N^{(0)}>0$ and boundary conditions. Let their
  external potentials $V$ and $V'$ be not diffusion equivalent, and
  denote by $l\in\N_0$ the smallest integer for which $\nabla d_V^{(l)}
  \not\equiv 0$.
  If
  \begin{align}
    \oint_{\partial\Omega} \rho(x,0) d_V^{(l)}(x) \nabla_x d_V^{(l)}(x) dS(x) = 0 ,
    \label{eq:divergecriterion}
  \end{align}
  then the two resulting densities are different, i.e., $\rho
  \not\equiv\rho'$. 
\end{theorem}

\begin{proof}
  By condition (A2), all time derivatives of $\rho$ and $\rho'$ exist. Thus if 
  \begin{align}
    \partial_t^{k} \left[\rho(x,t)-\rho'(x,t)\right]\Bigr|_{t=0} \not\equiv 0,
    \label{eq:somek}
  \end{align}
  for some $k\in\N$, then the densities are different.
  Consider the case $k=l+1$.
  According to \eqref{eq:toshow} from Lemma~\ref{lemma}, it suffices to show that
  \begin{align}
    \nabla_x \left[ \rho(x,0) \nabla_x d_V^{(l)}(x) \right] \not\equiv 0.
    \label{eq:elliptic-1}
  \end{align}
  Suppose, on the contrary, that $d_V^{(l)}(x)\not\equiv 0$ is a non-trivial solution
  of the PDE:
  \begin{align}
    \nabla_x \left[ \rho(x,0) \nabla_x d_V^{(l)}(x) \right] \equiv 0 .
    \label{eq:elliptic}
  \end{align}
  Now, consider the following integral
  \begin{align}
    \int d_V^{(l)}(x) \nabla_x &\left[ \rho(x,0) \nabla_x d_V^{(l)}(x) \right] \,dx \\
    =\; &  
    -\int  \rho(x,0) \left[ \nabla_x d_V^{(l)}(x) \right]^2 \,dx
    + \oint_{\partial\Omega} \rho(x,0) d_V^{(l)}(x) \nabla_x d_V^{(l)}(x) \,dS(x)
    \label{eq:trick}
  \end{align}
  using partial integration. By our assumption, the surface term vanishes,
  and we obtain
  \begin{align}
    \int d_V^{(l)}(x) \nabla_x \left[ \rho(x,0) \nabla_x d_V^{(l)}(x) \right]  \,dx
    &=
    -\int  \rho(x,0) \left[ \nabla_x d_V^{(l)}(x) \right]^2 \,dx.
  \end{align}
  Since $\nabla d_V^{(l)} \not\equiv 0$ and $\rho(x,0)>0$, the right-hand side is strictly negative.
  Therefore, the integrand on the left-hand side cannot vanish for all
  $x\in\Omega$, which in turn implies \eqref{eq:elliptic-1}. 
  
\end{proof}

From Theorem~\ref{thm:surface}, we can distinguish different cases of
uniqueness based on the behavior of $\rho$ close to the boundary. In
fact, condition~\eqref{eq:divergecriterion} could allow for diverging
densities if the gradient of the potential difference vanishes fast
enough. We will discuss such cases in the next section.

Now, we consider the physically-intuitive no-flux boundary condition
from \eqref{eq:no-one-body-flux} and show that it guarantees condition
\eqref{eq:divergecriterion} from Theorem~\ref{thm:surface} for any $V$
and $V'$ which are not diffusion equivalent. This proof allows for an
explicit uniqueness theorem.

\begin{theorem}\label{thm:no-normal-flux} 
  For an analytically-accessible many-body system with initial condition 
  $P_N^{(0)}>0$ and no normal flux $j_{\perp}$ at the boundary $\partial\Omega$, 
  the density $\rho$ uniquely determines the external potential $V$ (up to 
  diffusion equivalence). 
\end{theorem}

\begin{proof} As discussed above, we want to show that the
  density--potential mapping is unique, i.e., $\rho\equiv\rho'$ implies
  $V\sim V'$, by proving the contrapositive, i.e., if $V$ and $V'$ are
  not diffusion equivalent, then $\rho(x,t)$ must differ from
  $\rho'(x,t)$ for some $x\in\Omega$ and $t>0$. 
  
  Let $V\not\sim V'$ and $l$ denote the smallest integer for which
  $\nabla d_V^{(l)} \not\equiv 0$, as in Theorem~\ref{thm:surface}. To
  show that $\rho\not\equiv\rho'$, we have to show that condition
  \eqref{eq:divergecriterion} holds. Since the normal flux is zero for
  the two systems, subtracting \eqref{eq:normalflux} yields
  \begin{align}
    -Dn(x) \nabla_{x}[\rho(x,t)-\rho'(x,t)] &-D \beta n(x)[\rho(x,t)\nabla_x V(x,t) -\rho'(x,t)\nabla_x V'(x,t)] \\
    & +D\beta n(x)[\rho(x,t) -\rho'(x,t)] R(x,t) \\
    & +D\beta n(x)\int [\rho_{2}(x,y,t)-\rho_{2}'(x,y,t)]K_2(x,y) \,dy = 0
  \end{align}
  for all $x\in\partial\Omega$.
  By applying $l$ subsequent time derivatives and \eqref{eq:vanish}, we get
  \begin{align}
    n(x)\partial_t^l[\rho(x,t)\nabla_x V(x,t) -\rho'(x,t)\nabla_x V'(x,t)]\Bigr|_{t=0} = 0,
  \end{align}
  and 
  \begin{align}
    n(x)\rho(x,0)\nabla_x d_V^{(l)}(x) \equiv 0
    \label{eq:boundary_1}
  \end{align}
  follows by the same argument as in \eqref{eq:product}.
  
  Since $\partial^l_t V(x,t)\bigr|_{t=0}$ and $\partial^l_t
  V'(x,t)\bigr|_{t=0}$ are continuously extensible to $\partial\Omega$,
  and since $\Omega$ is bounded, $d_V^{(l)}$ is bounded on
  $\partial\Omega$. Thus, we have obtained a weaker property:
  \begin{align}
    n(x)\rho(x,0)d_V^{(l)}(x)\nabla_x d_V^{(l)}(x) = 0
    \label{eq:boundary_2}
  \end{align}
  for all $x\in\partial\Omega$. Therefore the integrand of equation 
  \eqref{eq:divergecriterion} vanishes, which proves the assertion.
\end{proof}

\begin{remark}\label{remark:boundarycondition} 
  Theorem~\ref{thm:no-normal-flux} can be readily extended to other
  boundary conditions that prescribe the normal flux at $\partial
  \Omega$, provided that the boundary term in
  \eqref{eq:reducedSmoluchowski} remains zero. However, when applying
  such conditions to our general $N$-body framework, where each particle
  is conserved according to the Smoluchowski equation, one must
  carefully verify that these conditions are physically meaningful and
  well-defined. A prominent example for simulations could be periodic
  boundary conditions.
\end{remark}

Theorem~\ref{thm:no-normal-flux} together with Remark
\ref{remark:boundarycondition} capture a generic case where the density
and the normal flux at the boundary together uniquely specify the
external potential and hence all higher-order correlations. This
assertion is consistent with the PFT framework, where the density and
current together yield a complete statistical description of a
time-dependent many-body system~\cite{schmidt_power_2013,
schmidt_power_2022}. However, our result is less restrictive since we
only need to fix the normal flux at the boundary. In fact, an
alternative formulation of our uniqueness theorem is possible, which
only involves the current without further assumption on the boundary
condition.

\begin{corollary}\label{cor:flux} 
  For an analytically-accessible many-body system with initial condition
  $P_N^{(0)} >0$, the one-body current $j$ uniquely determines the
  external potential $V$ (up to diffusion equivalence). 
\end{corollary}
\begin{proof} 
    Consider two potentials $V\not\sim V'$ and $l$ denote the smallest integer for 
    which $\nabla d_V^{(l)} \not\equiv 0$. Now let $\rho$ and $\rho'$ be the densities 
    that result with the external potentials $V$ and $V'$ with initial condition 
    $P_N^{(0)}$ and identical boundary conditions, which satisfying conditions (A1) and (A2). 
    By \eqref{eq:onebodycurrent}, we have
    \begin{align}
        j(x,t) -j'(x,t) =& -D\nabla_x \left[\rho(x,t)-\rho'(x,t)\right] - D\beta \rho(x,t) \nabla_x d_V(x,t) \\ 
        & - D\beta \left[\rho(x,t)-\rho'(x,t)\right] \nabla_x V'(x,t) + D\beta R(x,t) \left[\rho(x,t)-\rho'(x,t)\right] \\
        & + D\beta \int \left[\rho_{2}(x,y,t) - \rho'_{2}(x,y,t)\right]K_2(x,y)\, dy.
        \label{eq:difference_j}
    \end{align}
    Now consider 
    \begin{align}
        \partial_t^l \left[j(x,t)-j'(x,t)\right]\Bigr|_{t=0}  = -D \Bigg[ & \partial_t^l \left[\rho(x,t)-\rho'(x,t)\right]\Bigr|_{t=0} \\
        & + \beta \sum_{k=0}^l \binom{l}{k} \partial_t^{l-k} \rho(x,t)\Bigr|_{t=0} \nabla_x d_V^{(k)}(x) \\
        &  +\beta \sum_{k=0}^l \binom{l}{k} \partial_t^{l-k} \left[\rho(x,t)-\rho'(x,t)\right]\Bigr|_{t=0} \nabla_x \partial_t^k V'(x,t)\Bigr|_{t=0} \\
        &  -\beta \sum_{k=0}^l \binom{l}{k} \partial_t^{l-k} \left[\rho(x,t)-\rho'(x,t)\right]\Bigr|_{t=0} \partial_t^k R(x,t)\Bigr|_{t=0} \\
        &  -\beta \int \partial_t^l \left[\rho_{2}(x,y,t) - \rho'_{2}(x,y,t)\right]\Bigr|_{t=0} K_2(x,y)\, dy \Bigg]
        \label{eq:derivation_j}
    \end{align}
    From Lemma~\ref{lemma} and $\rho(x,0)>0$ follows 
    \begin{align}
        \partial_t^l \left[j(x,t)-j'(x,t)\right]\Bigr|_{t=0} & = -D\beta \rho(x,0) \nabla_x d_V^{(l)} \not\equiv 0 .
        \label{eq:derivation_j_2}
    \end{align}
    Thus, following the proof of Theorem~\ref{thm:surface}, we have
    $j\not\equiv j'$. 
\end{proof}

\begin{remark}
  \label{remark:explictimplicit}
  The general condition for uniqueness in Theorem~\ref{thm:surface}
  depends on the difference $d_V$ of the potentials, as expected from
  the arguments in quantum
  mechanics~\cite{runge_density-functional_1984,
  ruggenthaler_existence_2015}. Intuitively speaking, uniqueness of the
  potential $V$ that realizes $\rho$ is stated (only) in a
  ``neighborhood'' of $V$, i.e., within a space of admissible
  potentials.

  In contrast, Theorem~\ref{thm:no-normal-flux} and
  Corollary~\ref{cor:flux} make no explicit assumption on $d_V$.
  Instead, the conditions on the current (at the boundary) allow for a
  ``global'' statement because they exclude all potentials outside this
  ``neighborhood'' of $V$. (There is, of course, still the additional
  assumption of smooth potentials according to our method of proof
  following Runge and Gross~\cite{runge_density-functional_1984}.)
\end{remark}

Since we assume in condition (A1) in Definition~\ref{def:analyticity}
that the potential and its time derivatives are continuously extensible
to $\partial\Omega$, we exclude potentials that diverge at the boundary
$\partial\Omega$, but recall that hard walls, which are a common choice
for an external potential in physics, are here represented via our
bounded domain $\Omega$. Definition~\ref{def:analyticity} only restricts
any additional (time-dependent) potential within the domain. For
unbounded systems, it will be interesting to relax condition (A1) in
future work (e.g., via a limit for infinite system sizes).
  
Note that in the proof of Theorem~\ref{thm:no-normal-flux}, the only
explicit requirement is that the difference $d_V^{(l)}$ does not diverge
on $\partial\Omega$. Hence, it is promising to consider a larger space
of external potentials, in the sense of
Remark~\ref{remark:explictimplicit}, where the potentials $V$ and $V'$
may diverge at the boundary while their difference $d_V^{(l)}$ remains
finite on $\partial \Omega$ so that uniqueness is still preserved in a
neighborhood of $V$.
Moreover, we expect that our results can also be generalized to
non-integrable density profiles to include such a simple case as the
homogeneous bulk. We illustrate these scenarios in the examples below.

\section{Loopholes to uniqueness}
\label{sec:examples}

Our explicit conditions for uniqueness imply cases where a violation
directly results in non-unique potentials, i.e., a density profile that
can be realized by two different external potentials (which are not
diffusion equivalent). According to Remark~\ref{remark:explictimplicit},
such non-unique potentials have to differ strongly, in the sense that
they violate \eqref{eq:divergecriterion}. Theorem~\ref{thm:surface}
obviously makes no direct statement about how to construct such
potentials. However, the proof structure provides cues on their
existence, as the uniqueness of the density-potential mapping is
essentially equivalent to having only trivial solutions of the elliptic
PDE~\eqref{eq:elliptic}. This insight allows us to specify a
constructive method of non-unique potentials for the ideal gas or, more
generally, under the adiabatic approximation~\cite{marconi_dynamic_1999,
archer_dynamical_2004, te_vrugt_classical_2020}. The latter assumes
that the two-body density $\rho_2(x,y,t)$ is a functional of the
one-body density $\rho(x,t)$ at the same time $t$. The resulting
``loopholes'' can also be illustrated for unbounded systems and
non-integrable densities.

A generic procedure to construct explicit loopholes to uniqueness that
violate \eqref{eq:divergecriterion} follows directly from
\eqref{eq:rho1} and the assumption that $\rho_2$ is a functional of
$\rho$. In that case, \eqref{eq:rho1} is a closed equation in the sense
that it only depends on $\rho$ and not on $\rho_2$, $\rho_3$, \dots~By
taking the difference for two systems with distinct external potentials,
we obtain a PDE similar to \eqref{eq:elliptic} but explicitly containing
$V'$ via $d_V(x,t)$, as defined in \eqref{eq:def_d}: 
\begin{align}
  \nabla_x[\rho(x,t)\nabla_x d_V(x,t)] = 0 \text{ for all } x\in\Omega, t\geq 0.
  \label{eq:nonuniquePDE}
\end{align}
Hence, we can add to the external potential $V$ a nontrivial solution of
\eqref{eq:nonuniquePDE} without changing the density $\rho$. However, by
definition of \eqref{eq:nonuniquePDE}, there will be a divergence-free
difference $j'(x,t) - j(x,t)= -D\beta\rho(x,t)\nabla_x d_V(x,t)$ of the
one-body currents in the two systems.
By construction, this distinct current $j'$ violates our condition for
uniqueness in Theorem~\ref{thm:surface} (and also contradicts our
requirement on the boundary flux in Theorem~\ref{thm:no-normal-flux} and
Remark~\ref{remark:boundarycondition}) so that, consistently, the
external potential can no longer be expected to be unique.

The construction of loopholes according to \eqref{eq:nonuniquePDE} can
be made even more explicit if the density is effectively one-dimensional
(i.e., homogeneous in all coordinates but one). In that case, the
elliptic PDE~\eqref{eq:nonuniquePDE} reduces to a Sturm-Liouville
problem~\cite{walter_gewohnliche_2000}, which can be solved explicitly:
\begin{align}
  d_V(x,t) = \int_{x_0}^{x} \frac{c(t)}{\rho(y,t)}\,dy
\label{eq:algo}
\end{align}
with $x_0\in\Omega$ and $c(t)$ being a purely time-dependent constant.
This additional external potential leads to the new current
$$j'(x,t)=j(x,t)-D\beta c(t).$$ Only $c(t)\equiv0$ results in a
potential $d_V$ that does not violate condition (A1), and in that case,
the additional current vanishes, $-D\beta c(t) = 0$.

The procedure to construct non-unique potentials for an ideal gas
\eqref{eq:algo} is carried out in the following example, inspired
by~\cite{maitra_demonstration_2001}, where we consider the bounded
domain $\Omega = [-\pi/2,\pi/2]\subset\R$. For the rest of this section,
we set all constants to one by our choice of units (i.e., without loss of
generality).

\begin{figure}[t]
  \centering
  \begin{subfigure}[b]{0.48\textwidth}
    \centering
    \includegraphics[width=\textwidth]{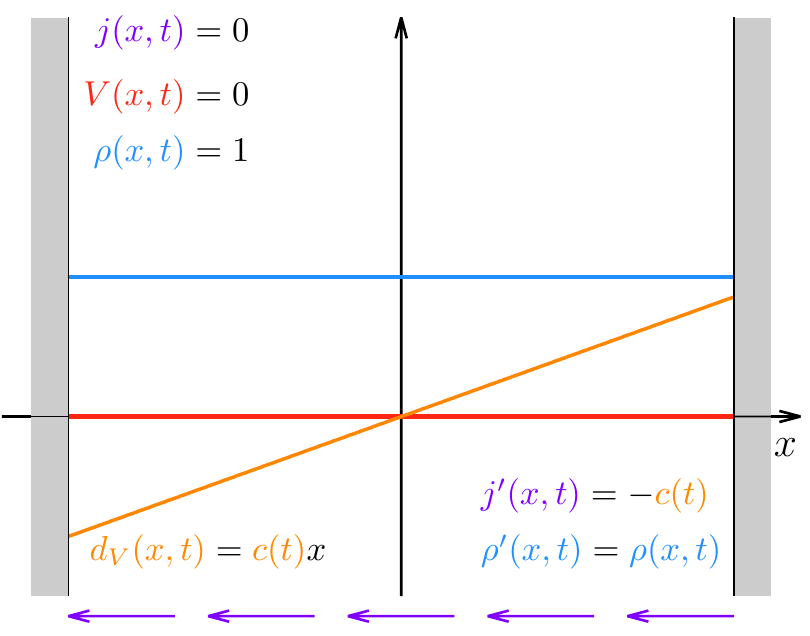}
    \caption{Bounded potentials}
    \label{fig_a_bounded}
  \end{subfigure}%
  \hfill%
  \begin{subfigure}[b]{0.48\textwidth}
    \centering
    \includegraphics[width=\textwidth]{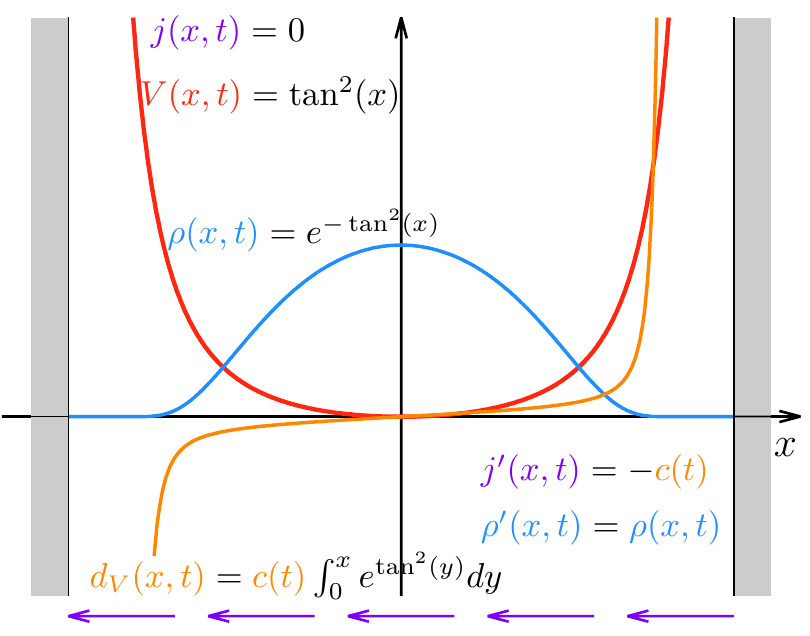}
    \caption{Diverging potentials}
    \label{fig_b_bounded}
  \end{subfigure}
  \caption{Schematics of non-unique potentials for bounded domains
  (illustrated as within hard walls): We start from the equilibrium
  solution $\rho(x,t)$ for non-interacting particles (a) without an
  external potential, i.e., $V(x,t)=0$; or (b) with a diverging potential
  $V(x,t)=\tan^2(x)$. Since the systems are in equilibrium,  $j(x,t)=0$.
  When a time-dependent potential $d_V(x,t)$ is added to the system (with
  $c(t)>0$), the density profile remains unchanged,
  $\rho'(x,t)\equiv\rho(x,t)$, but we obtain a spatially constant current
  $j'(x,t)\not\equiv0$, as indicated by the arrows at the bottom). Such a
  current requires a source and a sink at the boundary, in violation of
  our conditions in Theorem~\ref{thm:no-normal-flux}.
  }
  \label{fig_bounded}
\end{figure}

\begin{example}[Bounded systems]
  \label{example:bounded}
  Consider an ideal gas in (effectively) one dimension that is trapped
  in a bounded domain $\Omega = [-\pi/2,\pi/2]$, e.g., via hard walls,
  but without an additional potential inside the domain, i.e.,
  $V(x,t)\equiv 0$. Since the system is in equilibrium, there is no
  current $j(x,t)\equiv 0$, and the density $\rho(x,t)\equiv 1$ is
  constant in space and time; see Fig.~\ref{fig_bounded}~(a).

  If we add an external potential according to \eqref{eq:algo},
  \begin{align}
    d_V(x,t)=c(t)x,
    \label{eq:1d-homogeneous-loophole}
  \end{align}
  the potential gives rise to a current $j'(x,t)\equiv - c(t)$, which
  may vary in time, but which is spatially homogenous. Therefore, it
  preserves the homogeneous density $\rho'(x,t)\equiv 1$ even though
  particles get transported through the system.
  Obviously, such a current violates the no-flux boundary condition of
  hard walls (but it may be realized by periodic boundary conditions, as
  discussed below).

  Now, we consider another system, where the ideal gas is bounded by
  the diverging potential $V(x,t) = \tan^2(x)$. Its density is given by 
  $\rho(x,t)=e^{-\tan^2(x)}$. To obtain a homogeneous current that
  preserves the density profile, the added external potential $d_V$ has
  to diverge even stronger than the original potential, i.e.,
  \[d_V(x,t)=c(t)\int_0^x e^{\tan^2(y)}\,dy;\]
  see Fig.~\ref{fig_bounded}~(b).
\end{example}

The preceding example illustrates a more general principle. According to
Theorem~\ref{thm:surface}, if the density vanishes at the boundary, then
the nontrivial solution $d_V$ must diverge. For such a loophole to
uniqueness, it is important to check whether the boundary term of the
Smoluchowski equation in Theorem~\ref{thm:smoluchowski} actually
vanishes, e.g., using a constant number of particles and
Proposition~\ref{prop:YBG}. Otherwise, one may obtain inconsistent
results. For instance, if the example above were applied to a radially
symmetric potential, it would result in a radially symmetric current
that would not change the density profile, but particles would be
missing or accumulating at the boundary and Proposition~\ref{prop:YBG}
does not apply.

Further physically intuitive examples are available for an unbounded
domain, i.e., $\Omega=\R$. Note that whether $\Omega$ is (un-)bounded is
not essential to the question of uniqueness, as discussed above at the end of
Sec.~\ref{sec:uniqueness}.

\begin{figure}[t]
  \centering
  \begin{subfigure}[b]{0.48\textwidth}
    \centering
    \includegraphics[width=\textwidth]{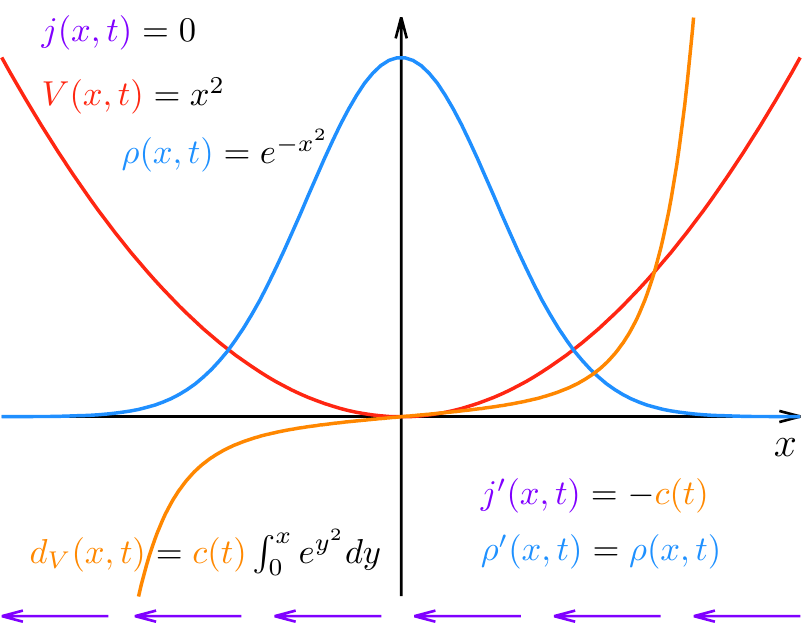}
    \caption{Exponentially diverging potentials}
    \label{fig_a}
  \end{subfigure}%
  \hfill%
  \begin{subfigure}[b]{0.48\textwidth}
    \centering
    \includegraphics[width=\textwidth]{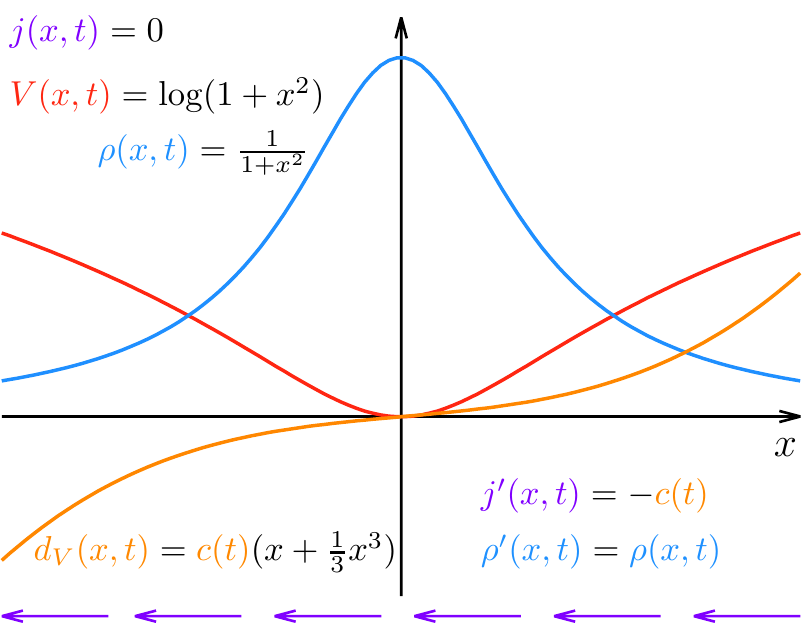}
    \caption{Polynomially diverging potentials}
    \label{fig_b}
  \end{subfigure}
  \caption{Schematics of non-unique potentials for \textit{unbounded} domains,
    analogous to Fig.~\ref{fig_bounded} but now with potentials
    (a) $V(x,t)=x^2$ and (b) $V(x,t)=\log(1+x^2)$.
  }
  \label{fig_unbounded}
\end{figure}

\begin{example}[Unbounded systems with localized density profiles]
  Consider an (effectively) one-dimensional ideal gas that is trapped in
  a harmonic potential, i.e., subject to an external potential
  $V(x,t)=x^2$ so that $\rho(x,t)=e^{-x^2}$; see
  Fig.~\ref{fig_unbounded}~(a). The system is again in equilibrium, and
  hence $j(x,t)\equiv 0$.

  An additional external potential according to \eqref{eq:algo},
  \begin{align}
    d_V(x,t)=c(t)\int_0^x e^{y^2}\,dy,
    \label{eq:expdivergence}
  \end{align}
  quickly diverges for $x\to\pm\infty$.  By construction, it results in
  a constant current that preserves the Gaussian density profile $\rho$.
  Intuitively speaking, the trick is that the external force is
  inversely proportional to $\rho$, i.e., it pulls stronger where the
  density is lower.

  In \eqref{eq:expdivergence}, the potential $d_V$ has to diverge
  exponentially fast to obtain the same density with different
  potentials. A slower divergence of $d_V$ suffices if the density
  profile has a heavy tail, e.g., $\rho(x,t)=1/(1+x^2)$ for
  $V(x,t)=\log(1+x^2)$, in which case $d_V(x,t)=O(x^3)$ for
  $x\to\pm\infty$, see Fig.~\ref{fig_unbounded}~(b).
\end{example}

Even though a constant $\rho$ defined on $\Omega=\R^d$ is, strictly
speaking, excluded from our setting (which assumes $N<\infty$), the
construction principle according to \eqref{eq:nonuniquePDE} and
\eqref{eq:algo} still works. It provides an obvious loophole to
uniqueness via a constant force.
\begin{example}[Homogeneous bulk density]
  Consider, analogous to the first case in
  Example~\ref{example:bounded}, an (effectively) one-dimensional ideal
  gas with $\rho(x,t)\equiv 1$. By adding the potential $d_V(x,t)=c(t)
  x$, we obtain a constant gradient that results in a constant current
  proportional to $c(t)$. Note that a no-flux boundary condition again
  excludes such a loophole to uniqueness.  Such a loophole works even
  for interacting systems.
\end{example}

Our procedure based on \eqref{eq:nonuniquePDE} and \eqref{eq:algo} can
also be applied to systems with periodic boundary conditions, where the
flux is not fixed at the periodic boundary. They can represent
simulations on the flat torus or periodic systems in unbounded domains.
  
\begin{example}[Periodic density profiles]
  Consider again an (effectively) one-dimensional ideal gas. If the
  system is defined on a flat torus, the additional potential $d_V(x,t)$
  is periodic. Hence it cannot be monotonically increasing as
  \eqref{eq:1d-homogeneous-loophole} would require; instead it has to be
  discontinuous (at the boundary). Alternatively, if the system is
  unbounded and described by a periodic density on $\R$, we again obtain
  a diverging potential on $\R$ in the non-unique case. In either
  setting, the additional potential violates our conditions for
  uniqueness. It generates a spatially constant change of the current.
  The forces remain bounded if the density is strictly positive. Such an
  example has already been numerically studied in simulations of
  interacting particles, going beyond the adiabatic approximation,
  within the framework of PFT and custom
  flow~\cite{de_las_heras_custom_2019, de_las_heras_perspective_2023}.
\end{example}

While the density does not uniquely specify the potential in any of the
aforementioned examples, the one-body current does so, as the underlying
principle of Corollary~\ref{cor:flux} still applies in each case (as
knowing the initial one-body density in our examples for an ideal gas
implies knowing the full initial $N$-body probability density). In this
strict mathematical sense, it is thus enough to know the current
(together with the initial density), even if it takes a simple form like
$j(x,t)=-D\beta c(t)$.

For general interaction potentials, $\rho_2$ will no longer be a
functional of $\rho$, and our simple procedure from
\eqref{eq:nonuniquePDE} and \eqref{eq:algo} breaks down. Instead, the
entire hierarchy of $n$-body correlations from
Theorem~\ref{thm:smoluchowski} has to be considered, in line with our
proof of uniqueness through Lemma~\ref{lemma}. This necessity to
properly account for superadiabatic
forces~\cite{fortini2014superadiabatic} can be easily overlooked in the
proof because it does not impose additional conditions on the uniqueness
theorem itself.

A practical way to resolve this problem of loopholes for systems with
super\-adiabatic forces is the so-called \textit{custom flow}
procedure~\cite{de_las_heras_custom_2019}. It employs many-particle
Brownian dynamics simulations to numerically determine for a given
density the potential (or, more generally, even non-conservative forces)
required to obtain a predefined flow profile; see
also~\cite{de_las_heras_perspective_2023}. 

\section{Physical implications}
\label{sec:outlookPHYS}

We have shown that the uniqueness of the potential can be inferred (i)
from the density alone, according to Theorem~\ref{thm:no-normal-flux},
see also Remark~\ref{remark:boundarycondition}, or (ii) from the current
alone, according to Corollary~\ref{cor:flux}. In both cases $P_N^{(0)}$
must be known since this initial condition is obviously required for a
well-defined behavior of an interacting many-particle system. 
What specifies appropriate boundary conditions at the $N$-body level in
general is related to the divergence theorem and the continuity equation
\eqref{eq:current_smoluN}. Let $P_N$ and $P_N'$ be different $N$-body
probability densities that evolve from the same initial condition but
differ in their time derivatives. According to the divergence theorem
and \eqref{eq:current_smoluN}, an integral of $\partial_t(P_N - P_N')$
over $\Omega^N$ is equivalent to an integral of $n(x^N) (J_N - J_N')$
over $\partial\Omega^N$. Therefore, different probability densities
correspond to different normal fluxes, which means that a well-posed
initial-value-boundary problem requires an appropriate condition on the
normal flux at the boundary. Surprisingly, even in the presence of
(known) non-conservative forces, the density suffices to formally
recover the $N$-body dynamics under physically meaningful boundary
conditions,  such as no flux, or, more generally, a fixed normal
one-body flux.

Closely related to the necessity of full $N$-body initial conditions is
our notion of a \textit{non-instantaneous} density-potential mapping in
the following sense. Uniqueness here means that $V(x,t)$ for a certain
point $x\in\Omega$ and at a certain time $t$ is determined only if the
values $\rho(y,s)$ (or $j(y,s)$) are given for all $y\in\Omega$ and $0
\leq s \leq t$. The mapping is thus not instantaneous (in the physical
  sense that only the quantities at time $t$ were known and not their
time derivatives). By requiring that our systems are analytically
accessible, we imply that $\rho$  is specified throughout; in that
sense, the potential at all times is determined by the one-body density
at all times. 

Importantly,  our results imply the existence of a theory solely based
on the one-body density (or current) which contains all higher-order
information on the dynamical system.
While our method of proof does not provide a constructive recipe of an
explicit theory, we can use the insight gained to answer the following
three essential questions on the adiabatic assumption underlying
standard DDFT and the general nature of an appropriate theoretical
framework.

\textit{Do our uniqueness theorems hold under the adiabatic approximation?} 
First note that establishing a uniqueness theorem or assessing the
performance of a density-based approximation are two separate problems.
In particular, the proof of uniqueness requires a well specified model
(and then no further approximations). Hence, we can distinguish two
different points of view:
(i) Adiabatic dynamics are utterly flawed for interacting systems and
the question is uninteresting since the density obtained from solving
the approximate PDE does not reflect the actual physical system.
Our rigorous statements are made for the exact dynamics. Hence the
question of applicability to an imprecise approximation is ill posed.
(ii) If we take the approximate adiabatic dynamics as the actual
mathematical model to be studied, then our theorems hold accordingly and
their proofs shorten considerably. 
By definition, the adiabatic approximation assumes that the information
contained in $\rho$ determines $P_N$ at all times. Thus, the proof of
Lemma~\ref{lemma} considerably simplifies to the argumentation required
for an ideal gas because the induction step for showing
\eqref{eq:vanish} only needs to be laid out for $n=1$ when we use the
functional chain rule.
Likewise, simple analytic counterexamples can be derived through
\eqref{eq:nonuniquePDE} under the adiabatic assumption or for an ideal
gas, see Sec.~\ref{sec:examples}.
In contrast, our hierarchical proof of Lemma~\ref{lemma} is essential
for the general problem and \eqref{eq:algo} cannot be used to construct
general loopholes because the actual dynamics of interacting particles
is always superadiabatic and the assumptions used to derive
\eqref{eq:nonuniquePDE} do not apply.

\textit{In how far are the known drawbacks of adiabatic DDFT reflected in our proof?} 
The adiabatic approximation only uses the instantaneous density, which
neither maps to the instantaneous potential nor provides accurate
dynamics in general.
To see this, we need to consider explicitly what happens when solving
Eq.\ \eqref{eq:rho1}. Our proof, with the help of Lemma~\ref{lemma},
explicitly takes into account the full hierarchy of $\rho_n$, which
ensures a proper transfer of the full information contained in
$P_N(\cdot,t)$ to $P_N(\cdot,t')$ at any later time $t'>t$, i.e., it
reflects the exact dynamics. In contrast, the adiabatic approximation
amounts to express the right-hand side of Eq.\ \eqref{eq:rho1} in terms
of a functional of the instantaneous one-body density $\rho$, such that
it does not contain the full information on its history. Formally, this
approximation only allows to propagate in time the information contained
in $\rho$ but not $P_N$, which is equivalent to only knowing the initial
condition $\rho(\cdot,0)$ instead of $P_N^{(0)}$. Strictly speaking, all
$P_N^{(0)}$ which are not given by a unique expression
$P_N^{(0)}[\rho(\cdot,0)]$ provided by the presumed underlying
equilibrium mapping (to an effective adiabatic potential 
\cite{de_las_Heras_conservation_2016}) cannot be represented at all. 
With this loss of generality, the density is not sufficient to properly
describe the $N$-body dynamics in the framework of adiabatic DDFT. 

\textit{How can our results contribute to going beyond adiabatic DDFT?}
While our Theorems assume the exact dynamics of the $N$-body system, any
practical theory must operate on a reduced set of variables, which
usually requires a closure to remove the implicit (hierarchical)
dependence on other variables.
Since our proofs are non-constructive, they do not directly translate
into an improved theory. Nevertheless, both our theorems and our proof
strategy help to clarify certain requirements.
Based on our results, an appropriate DDFT (i.e., a theory,  solely built
on the one-body density, that is fully closed in both space and time)
should, in general, use the full history of $\rho$ to model the
propagation in time. This indispensably requires a (generally unknown)
memory kernel. Generalized versions of DDFT utilizing such a flow kernel
have already been constructed for fluids under shear
\cite{brader_density_2011, scacchi_driven_2016,
scacchi_DDFT_sheared_2017}. Likewise, the framework of PFT formally
introduces memory through the excess dissipation functional  as an
extension of DDFT \cite{schmidt_power_2022}. Explicit PFT approximations
are provided in the form of a functional of both the density and the
current (or rather the velocity field), which is nonlocal in space and
time \cite{de_las_heras_gradient_2018,treffenstaedt_memoty_2020} and a
combination of PFT and neural networks can be used to characterize
memory effects \cite{Zimmermann_memoryPFT_2024}. 
In our induction argument for proofing Lemma~\ref{lemma}, the
information on the $N$-body initial condition $P_N^{(0)}$ is propagated
in time implicitly through the $n$-body densities $\rho_n$ and no
explicit account of the system history is needed (for analytically
accessible systems). Specifically, the knowledge of $\rho_n(\cdot,0)$
ensures a correct description of $\partial_t^{n-1}\rho(\cdot,t)$ at
$t=0$ and can thus account for differences in the external potentials up
to order $d_V^{(n-2)}$ in \eqref{eq:toshow}. This argument suggests a
minimal (or systematic) extension of adiabatic DDFT to include  memory
on a time-local level by explicitly considering the dynamics of $\rho_2$
(or further higher-order densities).
A practical recipe to do so is provided by superadiabatic-DDFT
\cite{tschopp_first-principles_2022, tschopp_superadiabatic_2023,
tschopp2024}, recently introduced by Tschopp and Brader, who arrive by a
different line of argument at the same conclusion. 
In~\cite{tschopp_superadiabatic_2023}, they argue that \textit{"within
superadiabatic-DDFT the flow history of the system is encoded in the
current value of the nonequilibrium two-body density, without the need
for a memory kernel"}. Following our argumentation in response to the
previous question, such a theory would formally propagate in time the
information contained in the initial $\rho_2(\cdot,0)$ and thus
represent these superadiabatic correlations, as
$P_N^{(0)}[\rho_2(\cdot,0)]$ is no longer enslaved by $\rho(\cdot,0)$,
and thus provide a very good description for systems interacting with
pair potentials, which evolve according to Eq.~\eqref{eq:rho1}. Indeed,
this superadiabatic-DDFT has been shown to  yield significant
improvement over standard adiabatic DDFT
\cite{tschopp_first-principles_2022, tschopp_superadiabatic_2023,
tschopp2024}.

To conclude, the drawbacks of adiabatic DDFT are not related to
violating uniqueness but can be understood and overcome when considering
the details of the proof related to superadiabatic correlations, see
Lemma~\ref{lemma}. Hence, while the density alone is often sufficient to
uniquely determine the generating forces, it is not a suitable basis
(with currently available methods) to construct a workable theory for
obtaining reliable predictions. In turn, the current approaches beyond
DDFT can stimulate further interesting mathematical questions. For
example, one could investigate whether our inclusion of
nonconservative forces provides additional insight into the PFT-based
structure-flow splitting of~\cite{de_las_heras_structure_2020}. An
alternative setting of our problem, along the lines of
superadiabatic-DDFT, is to find alternative conditions for uniqueness
using the two-body density instead of boundary currents.

\section{Outlook}
\label{sec:outlook}

In this paper, we derive explicit conditions for a well-defined
density--potential mapping under the assumptions of analytical
potentials and smooth densities. Given a well-posed initial value
boundary problem for the $N$-particle system, the external potential is
determined by the one-body density under a fixed boundary condition for
the normal flux or, alternatively, from the one-body current alone. Our
rigorous proof provides the foundation of dynamical density functional
theory (DDFT) and clarifies the relation to its recent extensions.

A useful extension will be to include non-integrable density profiles,
which should require a mere technical adaption of the setting and the
proof. Another open problem is to drop the condition of analytic
potentials. As mentioned above, a fixed-point approach as in
\cite{ruggenthaler_density-potential_2012, ruggenthaler_existence_2015}
could avoid this restriction. Such a generalization could also include
time-dependent diverging potentials (i.e., beyond static hard walls that
can already be encoded in our setting via no-flux boundary conditions).
Relatedly, our chosen restriction to a simply-connected domain and finite
interaction forces ensures that the particles can instantaneously reach each
point in the available space. Hence the density becomes positive in the whole
domain.
It will be interesting to loosen these assumptions in future work (similar to
Example~9.1 in~\cite{chayes_inverse_1984}) especially in view of scenarios
involving strong localization or caging of individual hard particles, where the
(canonical) density remains constantly zero in certain regions of space; such
scenarios are of particular interest in the context of
(D)DFT~\cite{tarazona_cavities_1997, lutsko_dft_canonical_2022, wittmann_caging_2021}. 

So far, we have also assumed the existence of a well-behaved solution
$P_N(x^N,t)$. A proof of existence could be constructed similarly to our
proof of uniqueness following ideas of van
Leeuwen~\cite{van_leeuwen_mapping_1999}.
The existence of a suitable potential will then require the solution to
an inhomogeneous PDE analogous to \eqref{eq:elliptic}. The resulting
conditions on the density and interaction potential should include, as a
special case, the known conditions for systems in
equilibrium~\cite{chayes_inverse_1984}. Similar questions have recently
been discussed in quantum mechanics~\cite{wrighton_problems_2023}.

Our method of proof has been chosen to provide physical intuitive
conclusions and may serve as the basis for investigating a broad range
of related problems. In particular, whenever we can derive a condition
of the form \eqref{eq:elliptic} the rest of the proof follows
analogously. Therefore, in future work, the rigorous search for unique
density--potential mappings should also be extended to more general
interactions and external influences, such as the following examples.
($i$)~Our strategy of proof can be straightforwardly applied to
particles with inertia, as originally considered by Chan and
Finken~\cite{chan_time-dependent_2005}. ($ii$)~Triplet or higher-order
potentials~\cite{lowen_role_1998, dobnikar_three-body_2004,
sammuller_inhomogeneous_2023}, as well as, time-dependent interactions
\cite{al-saedi_dynamical_2018, wittmann_Collective_mechano_2023,
Spannheimer_thesis_2024}, result in additional and more complex
contributions to the average interaction force but should not change the
structure of proof. ($iii$)~Similarly, our approach should be
generalizable to marked particles, where the marks may represent
different particle shapes, sizes, and
orientations~\cite{rex_dynamical_2007, wittkowski_dynamical_2011,
wittmann_fundamental_2016}.
($iv$)~A position-dependent mobility tensor should lead to a PDE
analogous to \eqref{eq:elliptic}, while $\rho(x,t)$ will be multiplied
with this tensor. This PDE will then be elliptic if and only if the
tensor is positive definite and symmetric, which is the case in models
for hydrodynamic interactions~\cite{rex_dynamical_2008,
menzel_dynamical_2016, bleibel_dynamic_2016}; see also
\cite{goddard_general_2012, goddard_well-posedness_2021}. Asymmetric
tensors that give rise to odd diffusivity~\cite{kalz_collisions_2022}
are, however, of growing interest in physics, as they, for example, can
be used to model the effect of Lorentz
forces~\cite{abdoli_nondiffusive_2020, abdoli_stationary_2020}. The
corresponding non-elliptic PDE might result in more involved conditions
for uniqueness.
($v$)~We further expect that our uniqueness theorems could also hold in
the additional presence of a (known) mechanism of self-propulsion that
constantly drives the system out of
equilibrium~\cite{wittkowski_dynamical_2011, menzel_dynamical_2016,
liverpool_steady-state_2020}. This activity generates an additional
non-conservative force field in \eqref{eq:force} that also depends on
the orientational dynamics and is thus not known \textit{a priori}.
Thus, this interplay needs to be separately accounted for in the proof
of Lemma~\ref{lemma} to ensure that superadiabatic properties of the
system are described appropriately, as in the case of known
non-conservative forces. It will be interesting to check whether the
uniqueness of the density-potential mapping would follow under
(slightly) modified conditions for these intrinsically non-equilibrium
systems.

\paragraph{Acknowledgements}
The authors are very grateful to Axel Gr\"unrock and Daniel de las Heras
for valuables pointers and the authors also thank Joseph Brader, Tobias
Kuna, Tanniemola B. Liverpool, and Salom\'ee Tschopp for stimulating
discussions. This work was funded by the Deutsche Forschungsgemeinschaft
(DFG, German Research Foundation), through the SPP 2265, under Grant
Nos.~WI 5527/1-2, KL 3391/2-2, and LO 418/25-2.
M.A.K.~also acknowledges funding and support by the Initiative and
Networking Fund of the Helmholtz Association through the Project
“DataMat.”

\paragraph{Data availability}
No datasets were generated or analyzed for this work. 

\paragraph{Competing interests}
All authors declare that they have no conflicts of interest. 



\end{document}